\documentclass[11pt, letterpaper]{amsart}

\usepackage{color}
\usepackage[english]{babel}
\usepackage[latin1]{inputenc}
\usepackage{times}
\usepackage[T1]{fontenc}
\usepackage{amsfonts}
\usepackage{epsfig,psfrag,latexsym}
\usepackage{amsmath, amssymb}
\usepackage{graphicx}
\usepackage{amsthm}

\newtheorem{theorem}{Theorem}[section]
\newtheorem{corollary}[theorem]{Corollary}
\newtheorem{lemma}[theorem]{Lemma}
\newtheorem{proposition}[theorem]{Proposition}
\newtheorem{condition}[theorem]{Condition}
\theoremstyle{definition}

\theoremstyle{remark}
\newtheorem{remark}[theorem]{Remark}

\numberwithin{equation}{section}

\title{Information geometry for approximate Bayesian computation}

\author{Konstantinos Spiliopoulos}
 \address{Department of  Mathematics and Statistics\\
 Boston University, Boston, MA, 02215}
 \email{kspiliop@math.bu.edu}

\thanks{The author would like to thank the Department of Mathematics and Statistics of UMASS Amherst for the hospitality and excellent working environment during the fall of 2016, when most of this work was initiated. In particular, the author would like to thank Markos Katsoulakis and Luc Rey-Bellet for a number of useful and fruitful discussions and suggestions on this project. This work was partially supported by  the National Science Foundation (NSF)
  CAREER award DMS 1550918.}

\date{\today}

\begin{document}

\begin{abstract}
The goal of this paper is to explore the basic Approximate Bayesian Computation (ABC) algorithm via the lens of information theory. ABC is a widely used algorithm in cases where the likelihood of the data is hard to work with or intractable, but one can simulate from it. We use relative entropy ideas to analyze the behavior of the algorithm as a function of the threshold parameter and of the size of the data. Relative entropy here is data driven as it depends on the values of the observed statistics. Relative entropy also allows us to explore the effect of the distance metric and sets up a mathematical framework for  sensitivity analysis allowing to find important directions which could lead to lower computational cost of the algorithm for the same level of accuracy. In addition, we also investigate the bias of the estimators for generic observables as a function of both the threshold parameters and the size of the data. Our analysis provides error bounds on performance for positive tolerances and finite sample sizes. Simulation studies complement and illustrate the theoretical results.
\end{abstract}

\maketitle



\section{Introduction}

The goal of this paper is to explore Approximate Bayesian Computation (ABC) methods via the lens of information theoretic criteria offering a more geometric point of view on such methods. ABC is a very popular likelihood-free Bayesian inference method in many fields of science ranging from population genetics to genetic evolutions and material science, see for example \cite{TavareGriffits1997,TanakaSisson2006,BlumTran2010,Bortot2007,Beaumont2002,WilkinsonSteiperEtAl2011} to name just a few.

In ABC, the goal is to do Bayesian inference on a parameter $\theta\in \mathbb{R}^{m}$ that has some prior distribution, say $f(\theta)$, given that we observe data $X\in\mathbb{R}^{n}$. The density $f(x|\theta)$ is usually either not available or not in a tractable form.  The data is usually summarized in form of appropriate statistics, $T(X):\mathbb{R}^{n}\mapsto\mathbb{R}^{q}$. The basic idea of ABC is to draw samples of $\theta$ from the prior distribution $f(\theta)$, then propagate the information forward producing simulated values of the sufficient statistic, say $T_{s}$ based on the distribution of $X|\theta$. Then if the simulated $T_{s}$ is close to the observed value $\tau^{*}=T_{o}=T(X)$ of the observed statistic, the sample $\theta$ is accepted, otherwise it is rejected. The procedure is then repeated a sufficiently large number of times allowing estimation of $\theta$ or of observables of interest.

ABC methods have been very popular, mainly due to their generality and simplicity in application. However, the generality and likelihood-free advantage that ABC has, comes with known issues on sensitivity in performance with respect not only to the distance metric $\|\cdot\|$ being chosen to compare simulated versus observed values, $T_{s}$ and $T_{o}$ respectively, but also on the choice of the threshold parameter $\varepsilon$ measuring whether the distance is large or not. In addition, the practical performance of ABC also depends on the choice of the observed statistics. One would like the statistics to be sufficient statistics, but that is usually hard to come up with in practice. It is also known that the basic ABC algorithm typically scales badly with the size of the data, say $n$, even though the use of sufficient statistics (when available) helps reduce the impact of big values of $n$. For these reasons, many clever modifications to the original algorithm have been proposed in the literature, including replacing the rejection mechanism by a kernel weighting the samples and doing localized regression \cite{Beaumont2002},  placing ABC in MCMC context \cite{MarjoramTavare2003}, Gibbs sampling in ABC \cite{WilkinsonSteiperEtAl2011}, using non-linear regression models in ABC \cite{BlumFrancois2010} and others. Excellent tutorials on the different variants of ABC methods are \cite{MarinPudloRobertRyder2012,TurnerVanZandt2012}.

Despite significant developments in methodological approaches, theoretical developments are still in their infancy. Notable exceptions to this are the recent works \cite{BarberVossWebster,BiauGuyader2015} where the behavior of the bias as a function of the threshold parameter for the Euclidean metric is explored and the expected cost of the basic version of the algorithm and  limiting properties are studied.

The focus of our paper is different. We take a geometric point of view based on studying the relative entropy between the underlying and the simulated posterior distribution of $\theta|T(X)$.  The relative entropy between a measure $P$ with respect to a measure $Q$ on a Euclidean space $\mathcal{X}$ is defined to be
\[
H(P|Q)=\int_{\mathcal{X}}\log\frac{dP(x)}{dQ(x)}dP(x)
\]
 when the integral exists. The relative entropy is not a measure but it is a divergence in that it satisfies  $H(P|Q)\geq 0$ and $H(P|Q)=0$ if and only if $P=Q$ almost everywhere. Information theoretic criteria have been successfully used in uncertainty quantification and sensitivity analysis, see for example \cite{ArampatzisKatsReyBellet2016,ArampatzisKats2014,DupuisKatsPantazisPlechac2016,KatsReyBelletWang2016,KatsPantazis2013}.

This geometric point of view focuses on the entire distribution rather than on specific statistics. At the same time, it gives us the flexibility to quantify the effect of the distance metric and of the threshold parameter. For example, we can consider the threshold parameter to be a vector, i.e., $\varepsilon\in\mathbb{R}^{q}_{+}$ instead of a scalar. As we show, the latter allows us to study sensitivity of the relative entropy with respect to the different components of the vector of observed statistics. In practice this means that if certain directions are insensitive then one can afford having larger values of threshold parameters $\varepsilon_{i}$ in such directions. The latter can increase the acceptance region, allowing ABC to run more efficiently.

In this paper we study how the relative entropy of the underlying measure, say $P$, with respect to the simulated measure, say $P^{\varepsilon}$, scales as a function of $\varepsilon$ for $|\varepsilon|$ small enough. Then we can choose $\varepsilon$ such that the acceptance region is maximized for a fixed level of tolerance for the relative entropy.

The contribution of the paper is fourfold. Firstly, in Section \ref{S:REgeneral} we compute the leading order term  of the relative entropy, characterizing explicitly its dependence on the size of the vector $T$, $q$, and the threshold vector $\varepsilon\in\mathbb{R}^{q}_{+}$. Secondly, as discussed, we allow $\varepsilon\in\mathbb{R}^{q}_{+}$ to be multidimensional. We compute the weighted effect in each direction $i=1,\cdots, q$ expressed in terms of the posterior distribution $f(\theta|T(x))$ and make the computation explicit for the exponential family in Section \ref{S:ExponentialFamily}.  We demonstrate in  specific examples that by exploiting the asymmetric effect of each direction, one is able to enlarge the acceptance region of ABC. In Section \ref{S:MeanRejectionRate} we compute the theoretical mean rejection rate of ABC in both cases of  $\varepsilon\in\mathbb{R}^{q}_{+}$ and of  $\varepsilon\in\mathbb{R}^{1}_{+}$ and we compare the two. We find that the mean rejection rate is typically smaller in the case of $\varepsilon\in\mathbb{R}^{q}_{+}$.

As we shall see in Sections \ref{S:Example} and \ref{S:Simulations} allowing a threshold parameter to be different for each statistic $\tau_{i}$ can have measurable effects in the performance of the algorithm. For example, in Figure \ref{F:EllipseBall}, we visually compare the acceptance regions for the ABC algorithm in the asymmetric case (i.e. acceptance region is an ellipse) and the symmetric case (i.e acceptance region is a ball) for the simulated example of Section \ref{S:Simulations}. The improvement of performance is because the first direction (in this example, the mean) has less significance as compared to the second direction (in this example, the variance) and thus by taking advantage of this we can allow for a larger acceptance region. The overall tolerance error is the same in both cases.

\begin{figure}[h!]
\begin{center}
\includegraphics[scale=1, width=8.0 cm, height=6.0 cm, angle=0]{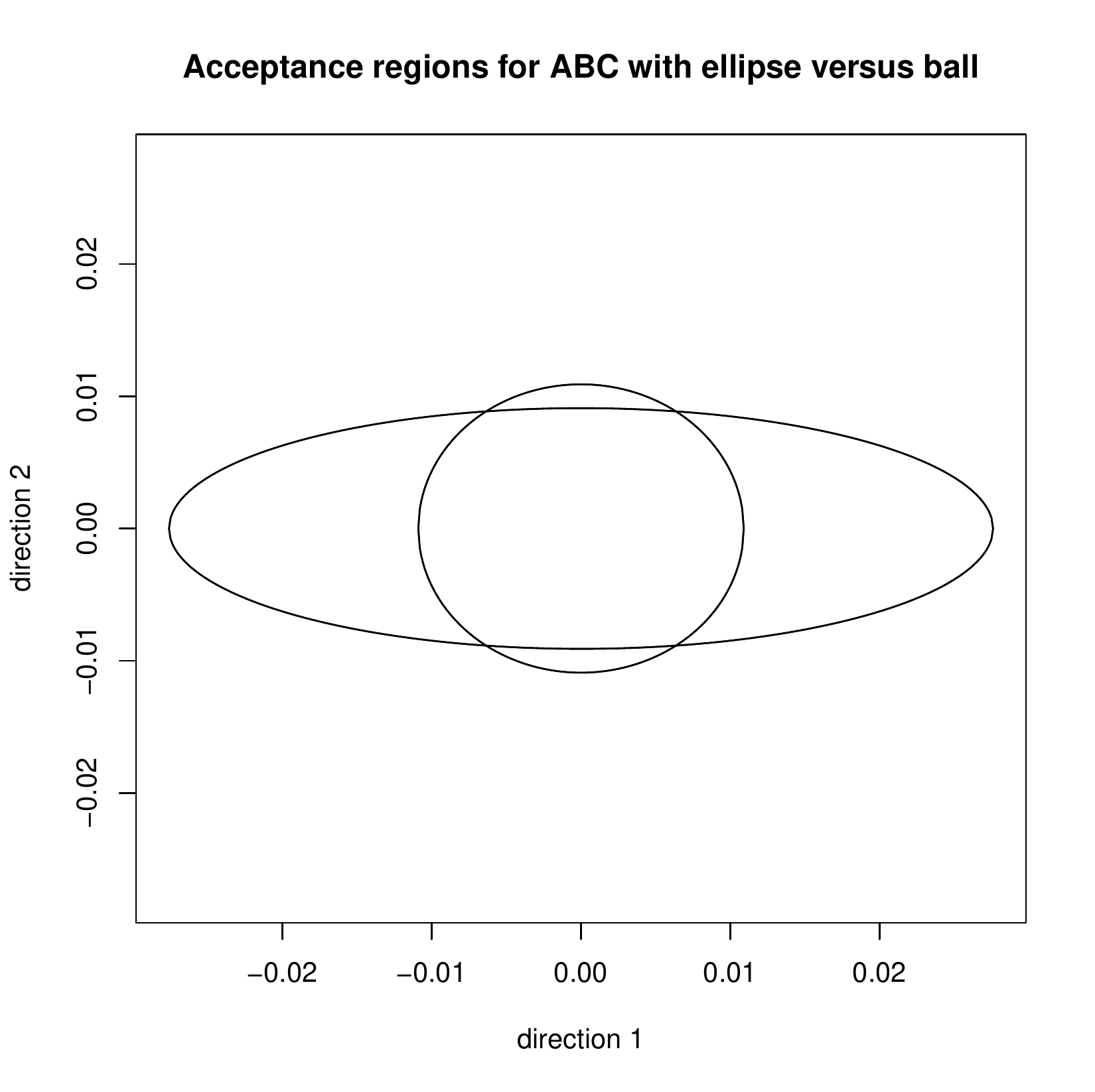}\hspace{2cm}
\caption{Indicative comparison of acceptance regions for ellipse and ball. Plots correspond to the data of Tables \ref{T:Ball} and \ref{T:Ellipse} of Section \ref{S:Simulations} for $\text{tolerance}=0.05$ and $n=1000$.}\label{F:EllipseBall}
\end{center}
\end{figure}


Of course, allowing the threshold parameter to be potentially different for each different statistic $\tau_i$, as opposed to forcing them to be the same, will naturally lead to larger acceptance regions. As we discuss in more detail in Section \ref{S:REgeneral} such choices (ellipse vs ball etc) correspond to special cases of distance metrics. Other choices of distance metrics are of course possible. In this paper we focus on explicitly characterizing the performance of ABC under the choice of the Euclidean metric being the distance metric via the lens of relative entropy. The relative entropy formulation gives us a mathematical framework to quantify the effect of the distance metric.

Third contribution of the paper is that we investigate in Section \ref{S:BiasCalculations} the behavior of the bias of estimators as a function of both $\varepsilon\in\mathbb{R}^{q}_{+}$  and of the size of the data $n$. We compute the weighted effect from each direction $i=1,\cdots,q$ on the bias.  We show that the weighted effects on both relative entropy and bias are based on the same weight functions.

Fourthly, in the case of exponential family and for sufficient statistics, we show that the typical order of the relative entropy is $n^{4}\sum_{i=1}^{q}\varepsilon_{i}^{4}/q$, whereas the typical order of the bias is $n^{2}\sum_{i=1}^{q}\varepsilon_{i}^{2}/q$. The latter naturally means that as $n$ increases, one should choose smaller values for $\varepsilon_{i}$. However, as we show in detail, this is a crude upper bound coming from the dominant term of the relative entropy and improvements are usually possible if one takes into account the specific effects of the weight functions. In particular, as we show in Section \ref{S:Example}, the relative entropy in many cases can be of much smaller order in $n$ and the bias of certain observables of interest may even be independent of $n$.

We remark here  that this inverse relation between $n$ and $|\varepsilon|$  is to be expected. After all, as $n\rightarrow\infty$ law of large numbers guarantees that sufficient statistics will converge to the true values of the corresponding parameters that they estimate. Hence, as $n$ gets larger, only samples that result in values for the simulated statistics that are closer to the values of the related observed statistics should be accepted. Essentially, this means that as $n$ increases, $|\varepsilon|$ should decrease as the variability in the estimation gets smaller. In this paper we quantify this inverse relation precisely.

Mathematically speaking, if we pick a level of tolerance, say $\text{tol}$, we study the following problem
\begin{align}
&\textrm{maximize }\varepsilon-\textrm{dependent acceptance region subject to } H(P|P^{\varepsilon})(\tau^{\ast})\leq\text{tol},\label{Eq:OptimizationProblem0}
\end{align}

where $P$ is the posterior distribution with density $f(\theta|T(X)=\tau^{*})$ and $P^{\varepsilon}$ is the simulated posterior distribution with density $f^{\varepsilon}(\theta|T(X)=\tau^{*})$.

It is easy to see though that (\ref{Eq:OptimizationProblem0}) is in general intractable. However, since in ABC the interest is in small values of $|\varepsilon|$ it makes sense to go to the limit as $|\varepsilon|\rightarrow0$. We prove that for $|\varepsilon|$ small enough one has the expansion
\begin{align}
H(P|P^{\varepsilon})(\tau^{\ast})&=\frac{1}{8}\frac{1}{(q+2)^{2}f^{2}(\tau^{\ast})}\mathbb{E}_{f(\theta|\tau^{\ast})}
\left[\left(\sum_{i=1}^{q}\varepsilon_{i}^{2} w_{i}(\tau^{\ast},\theta)\right)^{2}\right]+ \mathcal{O}(|\varepsilon|^{6}),\nonumber
\end{align}
where the weight $w_{i}(\tau^{\ast},\theta)$ can be computed explicitly and in the case of exponential family one can also see the exact dependence on the number of data points $n$. Notice that the relative entropy now depends on the observed value of $\tau^{\ast}$, i.e., it is data driven.

The latter expansion for the constraint means that if one sets a tolerance level, say tol, and requires that $H(P|P^{\varepsilon})(\tau^{\ast})=\text{tol}$ and that $|\varepsilon|$ is sufficiently small, then one should choose $\varepsilon_{i}$ such that
\begin{align}
&\textrm{maximize }\varepsilon-\textrm{dependent acceptance region subject to }\nonumber\\
&\qquad\frac{1}{8}\frac{1}{(q+2)^{2}f^{2}(\tau^{\ast})}\mathbb{E}_{f(\theta|\tau^{\ast})}
\left[\left(\sum_{i=1}^{q}\varepsilon_{i}^{2} w_{i}(\tau^{\ast},\theta)\right)^{2}\right]\leq\text{tol}.\label{Eq:OptimizationProblem}
\end{align}

The latter expression makes the problem tractable and it also makes it clear that the value of $\varepsilon_{i}$ depends on the value of the weight  factor $\mathbb{E}_{f(\theta|\tau^{\ast})}
 \left(w^{2}_{i}(\tau^{\ast},\theta)\right)$ in the $i^{\text{th}}$ direction and on the cross products $\mathbb{E}_{f(\theta|\tau^{\ast})}
 \left(w_{i}(\tau^{\ast},\theta)w_{j}(\tau^{\ast},\theta)\right)$ for $i\neq j$. This formulation allows us to exploit the asymmetric effect of different directions leading to potentially larger acceptance region for the basic version of the ABC algorithm. 
Notice also that (\ref{Eq:OptimizationProblem}) is a constrained optimization problem and in general is non-convex. We refer the interested reader to classical sources, such as \cite{Bertsekas2016, Boyd,Luenberger} for more on numerical solution to constrained optimization as well as non-convex problems.

The  paper is organized as follows. In Section \ref{S:REgeneral} we present the general mathematical framework that we consider and the relative entropy computations. We explore the dependence of the relative entropy on the magnitude of $|\varepsilon|$ and characterize the leading order of the expansion as $|\varepsilon|\rightarrow 0$. In Section \ref{S:ExponentialFamily} we make the computations specific for exponential family. This allows to see the dependence on the size of the data $n$, but it also gives a very intuitive interpretation of the weight factors as the variance of the square of the natural parameters in the representation of the density. In Section \ref{S:BiasCalculations} we calculate the asymptotics as $|\varepsilon|\rightarrow 0$ for the bias of the estimator for given observables characterizing the leading order term. We notice that the leading order term depends on the same weight functions $w_{i}(\tau^{*},\theta)$. In Section \ref{S:MeanRejectionRate} we compare analytically the conditional (on the observed statistics) mean rejections rates of the algorithm that uses the ellipse as acceptance region versus the algorithm that uses the ball.  Section \ref{S:Example} contains some examples to demonstrate that advantages can come from monitoring the weight factors and from considering different tolerance levels for different sufficient statistics. Section \ref{S:Simulations} contains the simulations studies that complement and illustrate the theoretical results. Conclusions and potential future work directions are presented in Section \ref{S:Conclusions}.

\section{Information-theoretic criteria for ABC}\label{S:REgeneral}
Let us assume that the parameter of interest is $\theta\in \Theta \subseteq \mathbb{R}^{m}$ and that we observe a  statistic of the data $X\in\mathbb{R}^{d}$, say $T(X)=\tau^{*}: \mathbb{R}^{d}\mapsto\mathbb{R}^{q}$.  In this paper and in order to be able to carry over the mathematics, we assume that we are working with sufficient statistics.

Consider a sampling distribution $f_{X|\Theta}(x|\theta)$ and assume a prior density $f_{\Theta}(\theta)$. In the basic ABC algorithm we decide to accept or reject a sample of $\theta$ from the prior density $f_{\Theta}(\theta)$ based on whether the statistic formed by the sampled data based on $f_{X|\Theta}(x|\theta)$, say $T_{s}=T(X)$, is within a level of tolerance, say $\varepsilon$, from the observed value of the statistic, $\tau^{\ast}=T_{o}$.

Apart from the choice of observed statistics,  the two main free parameters in ABC are (a): the distance metric $\|\cdot\|$  that we use  to measure the distance of $T_{s}$ from $T_{o}$, and (b): the tolerance parameter $\varepsilon$.

The ultimate goal is to compute, for a given $h(\theta)$, quantities of the form $\pi_{X}[h]=\mathbb{E}\left[h(\Theta)|X\right]$. Notice that if $T(\cdot)$ is a sufficient statistic, then we actually have  $\pi_{X}[h]=\mathbb{E}\left[h(\Theta)|T(X)\right]$.

Using the accept-reject mechanism just described, the estimator of the posterior distribution is
\begin{align}
\widehat{\pi_{X}[h]}_{\varepsilon,K}=\frac{1}{K}\sum_{i=1}^{K}h(\theta_{i}^{\varepsilon}),\label{Eq:Estimator}
\end{align}
where $K$ is the number of accepted samples.

Let us assume now that $T(X)$ and $\theta$ have  joint distribution with a well defined everywhere positive density which we denote by $f_{T,\Theta}(\tau,\theta)$. Likewise, let $f_{T}(\tau)$ be the density of the sufficient statistics $T$. We allow for $\varepsilon$ to be a vector in $\mathbb{R}^{q}_{+}$.

Let $P$ be the  distribution of $\Theta|T=\tau^{\ast}$ with density
\begin{equation}
f_{P}(\theta)=f_{\Theta|T}(\theta|\tau^{\ast})=\frac{f_{T,\Theta}(\tau^{\ast},\theta)}{f_{T}(\tau^{\ast})}.\nonumber
\end{equation}

Let $D_{\varepsilon}(\tau^{\ast})$ be the appropriate acceptance set and let $|D_{\varepsilon}(\tau^{\ast})|$ be its volume. The simulated data is distributed according to the conditional distribution of $\theta$ given $\tau\in D_{\varepsilon}(\tau^{\ast})$, say  $P^{\varepsilon}$. This means that in our case $P^{\varepsilon}$ will have density given by
\begin{equation}
f_{P^{\varepsilon}}(\theta)=f^{\varepsilon}(\theta|\tau^{\ast})=\frac{\frac{1}{|D_{\varepsilon}(\tau^{\ast})|}\int_{D_{\varepsilon}(\tau^{\ast})}f_{T,\Theta}(\tau,\theta)d\tau}{Z},\label{Eq:PerturbedDensity}
\end{equation}
where
\begin{equation}
Z=\frac{1}{|D_{\varepsilon}(\tau^{\ast})|}\int_{D_{\varepsilon}(\tau^{\ast})}f_{T}(\tau)d\tau,\nonumber
\end{equation}
is the necessary normalizing constant.  We remark here that $f_{P^{\varepsilon}}(\theta)$ also reflects the accept/reject mechanism that is in the heart of the ABC algorithm. Namely, if the simulated $\theta$ is such that the corresponding statistic $\tau$ is in $D_{\varepsilon}(\tau^{\ast})$, then the specific simulated value of $\theta$ is accepted. Otherwise, it is rejected.  We assume that  $\lim_{\varepsilon\rightarrow 0}D_{\varepsilon}(\tau^{\ast})=\{\tau^{\ast}\}$ for all $\tau^{\ast}\in\mathbb{R}^{q}$. In this paper we focus on Euclidean distances and for $\varepsilon\in \mathbb{R}^{q}_{+}$ we set
\begin{align}
D_{\varepsilon}(\tau^{\ast})&=\{\tau\in\mathbb{R}^{q}:\|\tau-\tau^{\ast}\|^{2}_{A(\varepsilon)}\leq 1\}, \label{Eq:DomainG}
\end{align}
where  $A(\varepsilon)$ is an $\varepsilon-$dependent positive definite matrix such that $\|\tau\|^{2}_{A(\varepsilon)}=\tau^{T}A^{-1}(\varepsilon)\tau$. We assume that for any value of the observed statistics $\tau^{\ast}$, $A(\varepsilon)$ is such that   $\lim_{|\varepsilon|\rightarrow 0}|D_{\varepsilon}(\tau^{\ast})|=0$.

There are of course many ways to choose the matrix $A(\varepsilon)$. One possibility, which is also the usual practice, is to set $A(\varepsilon)=\varepsilon^{2} I$ with $\varepsilon\in\mathbb{R}^{1}_{+}$ in which case we get a ball for acceptance region
\begin{align}
D_{\varepsilon}(\tau^{\ast})&=B_{\varepsilon}(\tau^{\ast})=\{\tau\in\mathbb{R}^{q}:\|\tau-\tau^{\ast}\|^{2}_{I}\leq\varepsilon^{2}\}. \label{Eq:Domain1}
\end{align}

Alternatively, we could pick $\varepsilon\in \mathbb{R}^{q}_{+}$, set $A(\varepsilon)=\text{diag}\left(\varepsilon_{1}^{2},\cdots,\varepsilon_{q}^{2}\right)$ and thus consider the ellipse
\begin{align}
D_{\varepsilon}(\tau^{\ast})=\left\{\tau\in\mathbb{R}^{q}:\sum_{i=1}^{q}\left|\frac{\tau_{i}-\tau^{\ast}_{i}}{\varepsilon_{i}}\right|^{2}\leq 1\right\}.\label{Eq:Domain2}
\end{align}

Other choices for the acceptance region $D_{\varepsilon}(\tau^{\ast})$ and distance metrics are of course possible. In this paper we will focus on analyzing the standard choice of a ball (\ref{Eq:Domain1}) and comparing it to the case of the ellipse (\ref{Eq:Domain2}).

For notational convenience we will be writing $f^{\varepsilon}(\theta|\tau)$ for the perturbed model and $f(\theta|\tau)$ for the true model given $\tau$. We compare the distributions $P$ and $P^{\varepsilon}$ via relative entropy
\begin{align}
H(P|P^{\varepsilon})(\tau^{\ast})&=\int\log\frac{dP}{dP^{\varepsilon}}dP=-\int\log\frac{f^{\varepsilon}(\theta|\tau^{\ast})}{f(\theta|\tau^{\ast})}f(\theta|\tau^{\ast})d\theta,\nonumber
\end{align}
which here is naturally data-driven and depends on the observed value of the sufficient statistic $\tau^{\ast}$.

As we discussed in the introduction, we choose the threshold parameters $\varepsilon_{i}$ based on the optimization problem (\ref{Eq:OptimizationProblem0}). Of course, such an optimization problem is difficult to solve as the relative entropy $H(P|P^{\varepsilon})(\tau^{\ast})$ is in general intractable. One can however consider the problem in the limit as $|\varepsilon|\rightarrow 0$. So our first task is to find a more specific expression for the first order approximation to $H(P|P^{\varepsilon})(\tau^{\ast})$ as $|\varepsilon|\rightarrow 0$.

Throughout the paper we will be assuming smoothness of the involved densities. In particular, we have the following assumption.
\begin{condition}\label{A:Regularity}
We assume that $f^{\varepsilon}(\theta|\tau), f(\theta|\tau), f_{T}(\tau)$ are smooth, bounded away from zero and that $f^{\varepsilon}(\theta|\tau^{\ast})$ is at least two times continuously differentiable with respect to $\varepsilon$.
\end{condition}

As $|\varepsilon|\downarrow 0$, it is easy to see that $f^{\varepsilon}(\theta|\tau^{\ast})$ converges to $f(\theta|\tau^{\ast})$.  In particular we have the following proposition.

\begin{proposition}\label{L:RelEntropy_GeneralExpansion}
Assume that Condition \ref{A:Regularity} holds and that the acceptance region $D_{\varepsilon}(\tau^{\ast})$ is given by (\ref{Eq:Domain2}). As $|\varepsilon|\rightarrow 0$ we have that
\begin{align}
H(P|P^{\varepsilon})(\tau^{\ast})
&=\frac{1}{8}\frac{1}{(q+2)^{2}f^{2}(\tau^{\ast})}\mathbb{E}_{f(\theta|\tau^{\ast})}
\left[\left(\sum_{i=1}^{q}\varepsilon_{i}^{2} w_{i}(\tau^{\ast},\theta)\right)^{2}\right]+ \mathcal{O}(|\varepsilon|^{6}),
 \label{Eq:ExpressionForTheEllipse}
\end{align}
where the weight function $w_{i}(\tau^{\ast},\theta)$ is
\begin{align}
w_{i}(\tau^{\ast},\theta)&=\frac{\partial_{\tau_{i}}^{2}f(\tau^{\ast},\theta)}{f(\theta|\tau^{\ast})}-
 \partial_{\tau_{i}}^{2}f(\tau^{\ast})\nonumber\\
 &=\frac{\partial_{\tau_{i}}^{2}f(\theta|\tau^{\ast})}{f(\theta|\tau^{\ast})}f(\tau^{\ast})+2
 \frac{\partial_{\tau_{i}}f(\theta|\tau^{\ast})}{f(\theta|\tau^{\ast})}
 \partial_{\tau_{i}}f(\tau^{\ast}).\label{Eq:weights}
\end{align}
\end{proposition}

\begin{proof}
Using Taylor series expansion on $f^{\varepsilon}(\theta|\tau^{\ast})$ with respect to $\varepsilon$, we have up to leading order for $|\varepsilon|$ sufficiently small
\begin{align}
H(P|P^{\varepsilon})(\tau^{\ast})&=\int\log\frac{dP}{dQ}dP=-\int\log\frac{f^{\varepsilon}(\theta|\tau^{\ast})}{f(\theta|\tau^{\ast})}f(\theta|\tau^{\ast})d\theta\nonumber\\
&=-\int\log\left(1+\frac{f^{\varepsilon}(\theta|\tau^{\ast})-f(\theta|\tau^{\ast})}{f(\theta|\tau^{\ast})}\right)f(\theta|\tau^{\ast})d\theta\nonumber\\
&=-\int\left(\frac{f^{\varepsilon}(\theta|\tau^{\ast})-f(\theta|\tau^{\ast})}{f(\theta|\tau^{\ast})}-\frac{1}{2}\left(\frac{f^{\varepsilon}(\theta|\tau^{\ast})-f(\theta|\tau^{\ast})}{f(\theta|\tau^{\ast})}\right)^{2}+\mathcal{O}(|f^{\varepsilon}-f|^{3})\right)f(\theta|\tau^{\ast})d\theta\nonumber\\
&=\frac{1}{2}\int\left(\left(\frac{f^{\varepsilon}(\theta|\tau^{\ast})-f(\theta|\tau^{\ast})}{f(\theta|\tau^{\ast})}\right)^{2}+\mathcal{O}(|f^{\varepsilon}-f|^{3})\right)f(\theta|\tau^{\ast})d\theta\nonumber\\
&=\frac{1}{2}\int\left(\left(\frac{\varepsilon^{T}\nabla_{\varepsilon}f^{\varepsilon}(\theta|\tau^{\ast})|_{\varepsilon=0}}{f(\theta|\tau^{\ast})}+\frac{1}{2}\frac{\varepsilon^{T}\nabla^{2}_{\varepsilon}f^{\varepsilon}(\theta|\tau^{\ast})|_{\varepsilon=0}\varepsilon}{f(\theta|\tau^{\ast})}\right)^{2}+\mathcal{O}(|\varepsilon|^{6})\right)f(\theta|\tau^{\ast})d\theta\nonumber\\
&=\frac{1}{2}\int\frac{\varepsilon^{T}\nabla_{\varepsilon}f^{\varepsilon}(\theta|\tau^{\ast})|_{\varepsilon=0}\left(\nabla_{\varepsilon}f^{\varepsilon}(\theta|\tau^{\ast})|_{\varepsilon=0}\right)^{T}\varepsilon}{f(\theta|\tau^{\ast})}d\theta\nonumber\\
&\qquad+\frac{1}{2}\int\frac{\varepsilon^{T}\nabla_{\varepsilon}f^{\varepsilon}(\theta|\tau^{\ast})|_{\varepsilon=0}\varepsilon^{T}\nabla^{2}_{\varepsilon}f^{\varepsilon}(\theta|\tau^{\ast})|_{\varepsilon=0}\varepsilon}{f(\theta|\tau^{\ast})}d\theta+\nonumber\\
&\qquad+\frac{1}{8}\int\frac{\left|\varepsilon^{T}\nabla^{2}_{\varepsilon}f^{\varepsilon}(\theta|\tau^{\ast})|_{\varepsilon=0}\varepsilon\right|^{2}}{f(\theta|\tau^{\ast})}d\theta+ \mathcal{O}(|\varepsilon|^{6})\nonumber\\
&=\frac{1}{2}\mathbb{E}_{f(\cdot|\tau^{\ast})}\left[\frac{\varepsilon^{T}\nabla_{\varepsilon}f^{\varepsilon}(\theta|\tau^{\ast})|_{\varepsilon=0}\left(\nabla_{\varepsilon}f^{\varepsilon}(\theta|\tau^{\ast})|_{\varepsilon=0}\right)^{T}\varepsilon}{f^{2}(\theta|\tau^{\ast})}\right]
+\nonumber\\
&\qquad+\frac{1}{2}\mathbb{E}_{f(\cdot|\tau^{\ast})}\left[\frac{\varepsilon^{T}\nabla_{\varepsilon}f^{\varepsilon}(\theta|\tau^{\ast})|_{\varepsilon=0}\varepsilon^{T}\nabla^{2}_{\varepsilon}f^{\varepsilon}(\theta|\tau^{\ast})|_{\varepsilon=0}\varepsilon}{f^{2}(\theta|\tau^{\ast})}\right]+\nonumber\\
&\qquad+\frac{1}{8}\mathbb{E}_{f(\cdot|\tau^{\ast})}\left[\frac{\left|\varepsilon^{T}\nabla^{2}_{\varepsilon}f^{\varepsilon}(\theta|\tau^{\ast})|_{\varepsilon=0}\varepsilon\right|^{2}}{f^{2}(\theta|\tau^{\ast})}\right]+ \mathcal{O}(|\varepsilon|^{6}).\label{Eq:RE_GeneralExpansion1}
\end{align}

To continue we need to study the contribution of the leading order term, which involves $\nabla_{\varepsilon}f^{\varepsilon}(\theta|\tau^{\ast})|_{\varepsilon=0}$ and $\nabla^{2}_{\varepsilon}f^{\varepsilon}(\theta|\tau^{\ast})|_{\varepsilon=0}$.  We consider the case where
$D_{\varepsilon}(\tau^{\ast})$ is an ellipse, as given by (\ref{Eq:Domain2}). 

By Lemma \ref{L:ZeroDerivative} in the appendix we have that $\nabla_{\varepsilon}f^{\varepsilon}(\theta|\tau^{\ast})|_{\varepsilon=0}=0$.  In particular as $|\varepsilon|\downarrow 0$, we then obtain that
\begin{align}
H(P|P^{\varepsilon})(\tau^{\ast})&=\frac{1}{8}\mathbb{E}_{f(\theta|T=\tau^{\ast})}\left[\frac{\left|\varepsilon^{T}\nabla^{2}_{\varepsilon}f^{\varepsilon}(\theta|\tau^{\ast})|_{\varepsilon=0}
\varepsilon\right|^{2}}{f^{2}(\theta|\tau^{\ast})}\right]+ \mathcal{O}(|\varepsilon|^{6}).\nonumber
\end{align}
Then, Lemma \ref{L:ZeroDerivative} again, but for the second derivatives now, gives
\begin{align}
H(P|P^{\varepsilon})(\tau^{\ast})&=\frac{1}{8}\frac{1}{(q+2)^{2}f^{2}(\tau^{\ast})}\mathbb{E}_{f(\theta|\tau^{\ast})}
\left[\left(\sum_{i=1}^{q}\varepsilon_{i}^{2} w_{i}(\tau^{\ast},\theta)\right)^{2}\right]+ \mathcal{O}(|\varepsilon|^{6}),\nonumber
\end{align}
where the weight function $w_{i}(\tau^{\ast},\theta)$ is given by (\ref{Eq:weights}). This completes the proof of the proposition.
\end{proof}

For completeness purposes we mention that a simple calculation shows that for every $i=1,\cdots,q$
\[
\mathbb{E}_{f(\theta|T=\tau^{\ast})}\left[w_{i}(\tau^{\ast},\theta)\right]=0.
\]

 Let us now comment on these results.
\begin{remark}\label{R:remark1}
\begin{enumerate}
\item{Notice that the formula for the ellipse, i.e., (\ref{Eq:ExpressionForTheEllipse})-(\ref{Eq:weights}), immediately reduces to the formula for the ball, when $\varepsilon_{1}=\cdots=\varepsilon_{q}=\varepsilon$. In particular, in the case of a ball we have the symmetric expression for the leading order term
    {\small\[
    \mathbb{E}_{f(\theta|T=\tau^{\ast})}\left[\frac{\left|\varepsilon^{T}\nabla^{2}_{\varepsilon}f^{\varepsilon}(\theta|\tau^{\ast})|_{\varepsilon=0}
\varepsilon\right|^{2}}{f^{2}(\theta|\tau^{\ast})}\right]=\varepsilon^{4}\frac{1}{(q+2)^{2}f^{2}(\tau^{\ast})}\mathbb{E}_{f(\theta|T=\tau^{\ast})}
\left[\frac{\left(\Delta_{\tau}f(\tau^{\ast},\theta)-
\frac{f(\tau^{\ast},\theta)}{f(\tau^{\ast})} \Delta_{\tau}f(\tau^{\ast})\right)^{2}}{f^{2}(\theta|\tau^{\ast})}\right]
    \]}
    We remark that in this case the effect of the individual weight is not being seen by the relative entropy.
    }
\item{Formula (\ref{Eq:ExpressionForTheEllipse}) makes clear that the typical dependence of the leading order term of the relative entropy on the size of the vector of sufficient statistics $q$ is of the order one.} 
\end{enumerate}
\end{remark}

Proposition \ref{L:RelEntropy_GeneralExpansion} quantifies the effect of choosing  $\varepsilon_{i}$ differently for different directions $i$.  
Depending on the weighted effects of $\mathbb{E}_{f(\theta|\tau^{\ast})}
 \left(w^{2}_{i}(\tau^{\ast},\theta)\right)$ and of the correlations $\mathbb{E}_{f(\theta|\tau^{\ast})}
 \left(w_{i}(\tau^{\ast},\theta)w_{j}(\tau^{\ast},\theta)\right)$ one may be able to obtain information in regards to which of the directions  $i=1,\cdots, q$ are more sensitive. We will see specific examples where this is true in Sections \ref{S:Example} and \ref{S:Simulations}.

\section{The case of exponential family}\label{S:ExponentialFamily}

In this section we focus on the case where $f(x|\theta)$ belongs in the exponential family and our purpose is twofold. Exponential family is a large class of distributions and its specific structure will allow us to get more explicit formulas and also see the effect that the size of the data, $n$, has on the behavior of the relative entropy.

We compute the weights $w_{i}(\tau^{\ast},\theta)$ given by (\ref{Eq:weights}) for an exponential family. Due to the specific structural form that the  involved densities have, we are able to get detailed formulas for $w_{i}(\tau^{\ast},\theta)$. In addition, the size of the data, $n$, appears in a clear way. In particular, in the case of an exponential family, the weights $w_{i}(\tau^{\ast},\theta)$ are of order $n^{2}$ which also implies that the  expansion of relative entropy in $\varepsilon$ is typically correct if the product $n|\varepsilon|$ is small. Note that the dependence on $n$ is also related to the fact that the relative entropy analysis here  compares the distribution of the model $P$ with respect to the model $P^{\varepsilon}$.

It is important to remark here though that in ABC the size of the data $n$ is fixed and one can only change $\varepsilon$. We  show that for fixed $n$ and sufficiently small $|\varepsilon|$, the leading order term of relative entropy will typically scale as $n^{4}\sum_{i=1}^{q}\varepsilon^{4}_{i}/q$. However, as we shall demonstrate in a specific example in Section \ref{S:ExmapleNormal}, this is only a crude generic upper bound and improvements are in general possible.

Since we assume that we have an exponential family we can write that
\begin{align}
f(x|\theta)=\zeta(x)\exp\{\sum_{i=1}^{q}\eta_{i}(\theta)S_{i}(x)-A(\theta)\}.\nonumber
\end{align}

Hence the joint distribution for $n$ i.i.d data $X_{1},\cdots, X_{n}$ is given by
\begin{align}
f(x_{1},\cdots,x_{n}|\theta)=\prod_{j=1}^{n}\zeta(x_{j})\exp\{\sum_{i=1}^{q}\eta_{i}(\theta)Y_{i}(x)-n A(\theta)\},\nonumber
\end{align}
where we have denoted the sufficient statistics by $Y_{i}(x)=\sum_{j=1}^{n}S_{i}(x_{j})$. It is then well known that the distribution of the vector $(Y_{1}(x),\cdots,Y_{q}(x))$ belongs in the exponential family too and the corresponding joint probability density function is given by
\begin{align}
f(y_{1},\cdots,y_{q}|\theta)=R(y_{1},\cdots,y_{q})\exp\{\sum_{i=1}^{q}\eta_{i}(\theta)y_{i}-n A(\theta)\},\nonumber
\end{align}
for an appropriate function $R(y_{1},\cdots,y_{q})$, whose form depends on the specifics of the distribution at hand. Instead of the sufficient statistics  $(Y_{1}(x),\cdots,Y_{q}(x))$ we would be interested in the vector $(T_{1}(x),\cdots,T_{q}(x))$ where for each $i=1,\cdots,q$,  $T_{i}(x)=\frac{Y_{i}(x)}{n}$. This is done in order to single out the effect of the size $n$ of the data set. A simple transformation then shows that the density of $(T_{1}(x),\cdots,T_{q}(x))$ is given by
\begin{align}
f(\tau_{1},\cdots,\tau_{q}|\theta)=n^{q}R(n\tau_{1},\cdots,n\tau_{q})\exp\{n\sum_{i=1}^{q}\eta_{i}(\theta)\tau_{i}-n A(\theta)\}.\nonumber
\end{align}

Now, as it is known, the natural conjugate prior for $\Theta$ takes the form
\begin{align}
f(\theta|\rho,\kappa)= Z(\rho,\kappa)\exp\{\sum_{i=1}^{q}\eta_{i}(\theta)\rho_{i}-\kappa A(\theta)\},\nonumber
\end{align}
where $Z(\rho,\kappa)$ is the appropriate normalizing factor. Then a relatively straightforward computation shows that the posterior distribution of $\Theta$ takes the form
\begin{align}
f_{\Theta|T}(\theta|\tau_{i},i=1,\cdots, q)= Z(\rho+n\tau,\kappa+n)\exp\{\sum_{i=1}^{q}\eta_{i}(\theta)(\rho_{i}+n \tau_{i})-(\kappa+n) A(\theta)\}.\label{Eq:PosteriorExponential}
\end{align}

Essentially, due to conjugancy,  we move from the prior to the posterior by setting $\rho\mapsto \rho + n \tau$ and $\kappa\mapsto \kappa+n$. The marginal likelihood of the data for $\tau=(\tau_{1},\cdots,\tau_{q})$ becomes
\begin{align}
f_{T}(\tau)=n^{q}R(n\tau_{1},\cdots,n\tau_{q})\frac{Z(\rho,\kappa)}{Z(\rho+n\tau,\kappa+n)}.\label{Eq:MarginalExponential}
\end{align}

Next goal is to make Proposition \ref{L:RelEntropy_GeneralExpansion} precise for the case of exponential family. We start with a special case to provide some intuition. The general case is presented in Subsection \ref{SS:WeightFactorsGeneralCase}.
\begin{lemma}\label{L:RelEntropySpecialForm1}
Let us  assume that we are in the special case where  $q=1$, that $n\varepsilon$ is sufficiently small and let us also assume that the function $R(n\tau)=\text{constant}$. Assume that the acceptance region $D_{\varepsilon}(\tau^{\ast})$ is given by (\ref{Eq:Domain1}) (or equivalently by (\ref{Eq:Domain2}) since $q=1$). Then, we have that
\begin{align}
H(P|P^{\varepsilon})(\tau^{\ast})&= (n\varepsilon)^{4}\frac{1}{72} \text{Var}_{f(\theta|\tau^{\ast})} \left(\eta^{2}(\theta)\right)+ \mathcal{O}(|n\varepsilon|^{6}).\label{Eq:REsmall_n_epsilon}
\end{align}
\end{lemma}
\begin{proof}

Let us denote $\tau_{\varepsilon}=\tau \varepsilon$. Based on the expression (\ref{Eq:PosteriorExponential}) we can write
\begin{align}
f^{\varepsilon}(\theta|\tau^{\ast})&=\frac{\int_{B_{1}(0)}f(\tau_{\varepsilon}+\tau^{\ast},\theta)d\tau}{\int_{B_{1}(0)}f(\tau_{\varepsilon}+\tau^{\ast})d\tau}=\frac{\int_{B_{1}(0)}f(\theta|\tau_{\varepsilon}+\tau^{\ast})f(\tau_{\varepsilon}+\tau^{\ast})d\tau}{\int_{B_{1}(0)}f(\tau_{\varepsilon}+\tau^{\ast})d\tau}\nonumber\\
&=f(\theta|\tau^{\ast})\frac{\int_{B_{1}(0)}R(n\tau_{\varepsilon}+n\tau^{\ast})e^{n \eta(\theta)\cdot \tau_{\varepsilon} }d\tau}{Z(\rho+n\tau^{\ast},\kappa+n)\int_{B_{1}(0)}R(n\tau_{\varepsilon}+n\tau^{\ast})Z^{-1}(\rho+n\tau_{\varepsilon}+n\tau^{\ast},\kappa+n)d\tau}. \nonumber
\end{align}

Hence, we obtain that the relative entropy satisfies
\begin{align}
H(P|P^{\varepsilon})(\tau^{\ast})&=-\int\log\frac{f^{\varepsilon}(\theta|\tau^{\ast})}{f(\theta|\tau^{\ast})}f(\theta|\tau^{\ast})d\theta\nonumber\\
&=-\int\log\frac{\int_{B_{1}(0)}R(n\tau_{\varepsilon}+n\tau^{\ast})e^{n \eta(\theta)\cdot\tau_{\varepsilon}}d\tau}{Z(\rho+n\tau^{\ast},\kappa+n)\int_{B_{1}(0)}R(n\tau_{\varepsilon}+n\tau^{\ast})Z^{-1}(\rho+n\tau_{\varepsilon}+n\tau^{\ast},\kappa+n)d\tau}f(\theta|\tau^{\ast})d\theta\nonumber\\
&=-\int\log\left(\int_{B_{1}(0)}R(n\tau_{\varepsilon}+n\tau^{\ast}) e^{n \eta(\theta)\cdot\tau_{\varepsilon}}d\tau\right)f(\theta|\tau^{\ast})d\theta\nonumber\\
&\qquad+\log\left(Z(\rho+n\tau^{\ast},\kappa+n)\int_{B_{1}(0)}R(n\tau_{\varepsilon}+n\tau^{\ast}) Z^{-1}(\rho+n\tau_{\varepsilon}+n\tau^{\ast},\kappa+n)d\tau\right). \nonumber
\end{align}

Under the assumptions of the lemma we then have
\begin{align}
H(P|P^{\varepsilon})(\tau^{\ast})&=-\frac{1}{6}(n\varepsilon)^{2} \mathbb{E}_{f(\theta|\tau^{\ast})} \eta^{2}(\theta)+\frac{2}{5!3} (n\varepsilon)^{4} \mathbb{E}_{f(\theta|\tau^{\ast})} \eta^{4}(\theta) \nonumber\\
&\quad+\frac{Z(\rho+n\tau^{\ast},\kappa+n)}{2}\left(\frac{1}{3}(n\varepsilon)^{2}\partial^{2}_{1}Z^{-1}(\rho+n\tau^{\ast},\kappa+n)+\frac{2}{5!}(n\varepsilon)^{4}\partial^{4}_{1}Z^{-1}(\rho+n\tau^{\ast},\kappa+n)
\right)\nonumber\\
&\quad-\frac{1}{72}(n\varepsilon)^{4}\left(Z(\rho+n\tau^{\ast},\kappa+n)\partial^{2}_{1}Z^{-1}(\rho+n\tau^{\ast},\kappa+n)\right)^2+ \mathcal{O}(|n\varepsilon|^{6}). \nonumber
\end{align}

Next, we notice that
\begin{align}
Z(\rho+n\tau^{\ast},\kappa+n)\partial^{(k)}_{1}Z^{-1}(\rho+n\tau^{\ast},\kappa+n)&= \mathbb{E}_{f(\theta|\tau^{\ast})}
 \eta^{k}(\theta). \nonumber
\end{align}

Hence, we can write
\begin{align}
H(P|P^{\varepsilon})(\tau^{\ast})&=-\frac{1}{6}(n\varepsilon)^{2} \mathbb{E}_{f(\theta|\tau^{\ast})} \eta^{2}(\theta)+\frac{2}{5!3} (n\varepsilon)^{4} \mathbb{E}_{f(\theta|\tau^{\ast})} \eta^{4}(\theta) \nonumber\\
&\quad+\left(\frac{1}{6}(n\varepsilon)^{2}\mathbb{E}_{f(\theta|\tau^{\ast})} \eta^{2}(\theta)+\frac{2}{5!}(n\varepsilon)^{4}\mathbb{E}_{f(\theta|\tau^{\ast})} \eta^{4}(\theta)
\right)-\frac{1}{72}(n\varepsilon)^{4}\left(\mathbb{E}_{f(\theta|\tau^{\ast})} \eta^{2}(\theta)\right)^2+ \mathcal{O}(|n\varepsilon|^{6}) \nonumber\\
&= (n\varepsilon)^{4}\frac{1}{72} \text{Var}_{f(\theta|\tau^{\ast})} \left(\eta^{2}(\theta)\right)+ \mathcal{O}(|n\varepsilon|^{6}),\nonumber
\end{align}
completing the proof of the lemma.
\end{proof}

Lemma \ref{L:RelEntropySpecialForm1} makes it clear that, at least in the case of exponential family, the typical order of the relative entropy is $(n\varepsilon)^{4}$ and that for a given small tolerance, say $\text{tol}$, one can choose
\[
(n\varepsilon)^{4}\frac{1}{72} \text{Var}_{f(\theta|\tau^{\ast})} \left(\eta^{2}(\theta)\right)\leq \text{tol}.\]
We also see that there is a weight factor that takes the form $\frac{1}{72}\text{Var}_{f(\Theta|\tau^{\ast})} \left(\eta^{2}(\theta)\right)$. In Section \ref{SS:WeightFactorsGeneralCase} we generalize the computation to the $q\neq 1$ case and for $R(\cdot)\neq \text{constant}$ based on the expression given by Proposition \ref{L:RelEntropy_GeneralExpansion}. Of course the expression derived in Section \ref{SS:WeightFactorsGeneralCase} coincides with (\ref{Eq:REsmall_n_epsilon}) with $q=1$ and for $R(\tau)=\text{constant}$, but it also shows that the individual weights may have non-trivial contribution, a point made clearer in Section \ref{S:Example}.

\subsection{Calculation of the weight factors in the general case.} \label{SS:WeightFactorsGeneralCase}
The goal of this section is to derive the analogous statement to Lemma \ref{L:RelEntropySpecialForm1} in the general case $q\neq 1$ and $R(n\tau)\neq\text{constant}$. We start with some assumptions. For a given function $h:\mathbb{R}^{q+1}\mapsto \mathbb{R}$, let $\partial_{i}h$ be the partial derivative in the $i^{\text{th}}$ direction.

\begin{condition}\label{A:ConditionR}
For each $k\in\mathbb{N}$ and for all $i=1,\cdots, q$ we have that $\partial^{(k)}_{i} \log R(n\tau^{\ast})$ is bounded uniformly on $n$ as $n\rightarrow\infty$.
\end{condition}

Condition \ref{A:ConditionR} is trivially true for the normal distribution for example. We also require
\begin{condition}\label{A:ConditionR2}
Define the function $g(\theta)=\sum_{i=1}^{q}\eta_{i}(\theta)\tau_{i}-A(\theta)$. Let us assume that the function g is smooth and has a maximum only at one interior non-degenerate critical point, say $\theta^{*}$, i.e., that $\nabla g(\theta^{*})=0$ and $\text{det}\nabla^{2}g(\theta^{*})<0$.
\end{condition}

\begin{lemma}\label{L:RelEntropy_ExpFamilyGeneral}
Assume that the acceptance region is given by by (\ref{Eq:Domain2}). Under Conditions \ref{A:Regularity}, \ref{A:ConditionR} and \ref{A:ConditionR2} we have as $|n\varepsilon|\rightarrow 0$,
{\small\begin{align}
&H(P|P^{\varepsilon})(\tau^{\ast})=\frac{n^{4}}{8}\frac{1}{(q+2)^{2}}\mathbb{E}_{f(\theta|\tau^{\ast})}
\left(\sum_{i=1}^{q}\varepsilon_{i}^{2}\left(\eta^{2}_{i}(\theta)-\mathbb{E}_{f(\theta|\tau^{\ast})} \eta^{2}_{i}(\theta)+2\partial_{i} \log R(n\tau^{\ast})\left(\eta_{i}(\theta)-\mathbb{E}_{f(\theta|\tau^{\ast})} \eta_{i}(\theta)\right)
\right)\right)^{2}\nonumber\\
&\qquad\qquad+ \mathcal{O}(|n\varepsilon|^{6}).
\label{Eq:ExpressionForTheEllipse2}
\end{align}}
\end{lemma}
The proof of this lemma is at the end of this section. A few remarks follow.

Lemma \ref{L:RelEntropy_ExpFamilyGeneral} shows that for  appropriately small $n |\varepsilon|$ and as long as the rest of the involved quantities are bounded uniformly on $n$, the leading order term of the relative entropy scales like $n^{4}\sum_{i=1}^{q}\varepsilon_{i}^{4}/q$. However, this is only a crude upper bound and is not using the individual properties of the weight factors $w_{i}$. The weight factors and the probability distribution under which the expectation is being calculated also depends on $n$, which implies that it is possible that for some $i$, the behavior is better than $n^{4}$. In particular,  in
 Section \ref{S:Example} we show that this is the case indeed for a few simple examples.

Let us also remark that if the function $R(\tau)$ is constant, (see Section \ref{S:ExmapleNormal} for such a case), then (\ref{Eq:ExpressionForTheEllipse2}) takes a simpler form. Indeed, we have the following corollary.
\begin{corollary}\label{C:RelEntropy_ExpFamilyGeneral}
Under the setting of Lemma \ref{L:RelEntropy_ExpFamilyGeneral} and assuming that $R(\tau)$ is constant, we obtain
\begin{align}
&H(P|P^{\varepsilon})(\tau^{\ast})=\frac{n^{4}}{8}\frac{1}{(q+2)^{2}}
\left(\sum_{i=1}^{q}\varepsilon_{i}^{4}\text{Var}_{f(\theta|\tau^{\ast})}\left(
\eta^{2}_{i}(\theta)\right)+\sum_{i,j=1, i\neq j}^{q}\varepsilon_{i}^{2}\varepsilon_{j}^{2}\text{Cov}_{f(\theta|\tau^{\ast})}\left(
\eta^{2}_{i}(\theta),\eta^{2}_{j}(\theta)\right)\right)
+ \mathcal{O}(|n \varepsilon|^{6})
\label{Eq:ExpressionForTheEllipse3}
\end{align}

In particular if also $q=1$ then we obtain the statement of Lemma \ref{L:RelEntropySpecialForm1}.
\end{corollary}

\begin{proof}[Proof of Lemma \ref{L:RelEntropy_ExpFamilyGeneral}]
We start by computing the weight factors $w_{i}(\tau^{\ast},\theta)$ given by (\ref{Eq:weights}).   By taking derivatives on (\ref{Eq:PosteriorExponential}) we obtain
\begin{align}
\partial_{\tau_{i}}f_{\Theta|T}(\theta|\tau)&=n\left(\partial_{i} \log Z(\rho+n\tau,\kappa+n)+\eta_{i}(\theta)\right)  f_{\Theta|T}(\theta|\tau)\nonumber\\
\partial^{2}_{\tau_{i}}f_{\Theta|T}(\theta|\tau)&=n^{2}\left(\partial^{2}_{i,i} \log Z(\rho+n\tau,\kappa+n)+\left(\partial_{i} \log Z(\rho+n\tau,\kappa+n)+\eta_{i}(\theta)\right)^{2}\right)  f_{\Theta|T}(\theta|\tau)\nonumber\\
\partial_{\tau_{i}}f_{T}(\tau)&=n\left( \partial_{i} \log R(n\tau)-\partial_{i} \log Z(\rho+n\tau,\kappa+n)\right)f_{T}(\tau).\label{Eq:DerivativeFactors}
\end{align}

Hence, by plugging these expressions into (\ref{Eq:weights}) we subsequently obtain
\begin{align}
w_{i}(\tau^{\ast},\theta)&=\frac{\partial_{\tau_{i}}^{2}f(\theta|\tau^{\ast})}{f(\theta|\tau^{\ast})}f(\tau^{\ast})+2
 \frac{\partial_{\tau_{i}}f(\theta|\tau^{\ast})}{f(\theta|\tau^{\ast})}
 \partial_{\tau_{i}}f(\tau^{\ast})\nonumber\\
 &=n^{2}\left[\partial^{2}_{i,i} \log Z(\rho+n\tau^{\ast},\kappa+n)+\left(\partial_{i} \log Z(\rho+n\tau^{\ast},\kappa+n)+\eta_{i}(\theta)\right)^{2}\right.\nonumber\\
 &\qquad\left.+2\left(\partial_{i} \log R(n\tau^{\ast})-\partial_{i} \log Z(\rho+n\tau^{\ast},\kappa+n)\right)\left(\partial_{i} \log Z(\rho+n\tau^{\ast},\kappa+n)+\eta_{i}(\theta)\right)\right]f_{T}(\tau^{\ast})\nonumber\\
 &=n^{2}\left[\partial^{2}_{i,i} \log Z(\rho+n\tau^{\ast},\kappa+n)-\left(\partial_{i} \log Z(\rho+n\tau^{\ast},\kappa+n)\right)^{2}+\eta^{2}_{i}(\theta)\right]f_{T}(\tau^{\ast})\nonumber\\
 &\quad+n^{2}\left[2\partial_{i} \log R(n\tau^{\ast})\left(\partial_{i} \log Z(\rho+n\tau^{\ast},\kappa+n)+\eta_{i}(\theta)\right)\right]f_{T}(\tau^{\ast}).\nonumber
\end{align}

Next, we notice that (\ref{Eq:PosteriorExponential}) implies that the $k^{th}$ derivative of $Z^{-1}(\rho+n\tau^{\ast},\kappa+n)=Z^{-1}(\rho_{1}+n\tau^{\ast}_{1},\cdots,\rho_{q}+n\tau^{\ast}_{q},\kappa+n)$ in the $i^{th}$ direction is for $i=1,\cdots, q$
\begin{align}
\partial^{(k)}_{i}Z^{-1}(\rho+n\tau^{\ast},\kappa+n)&=Z^{-1}(\rho+n\tau^{\ast},\kappa+n) \mathbb{E}_{f(\theta|\tau^{\ast})}
 \eta^{k}_{i}(\theta).\nonumber
\end{align}

Hence, we can actually write
\begin{align}
\partial^{2}_{i,i} \log Z(\rho+n\tau^{\ast},\kappa+n)-\left(\partial_{i} \log Z(\rho+n\tau^{\ast},\kappa+n)\right)^{2}&= -\mathbb{E}_{f(\theta|\tau^{\ast})} \eta^{2}_{i}(\theta),\label{Eq:Zcombination1}
\end{align}
and
\begin{align}
\partial_{i} \log Z(\rho+n\tau^{\ast},\kappa+n)&= -\mathbb{E}_{f(\theta|\tau^{\ast})} \eta_{i}(\theta).\label{Eq:Zcombination2}
\end{align}

The latter imply that the weight function takes the form
\begin{align}
w_{i}(\tau^{\ast},\theta) &=n^{2}\left[\eta^{2}_{i}(\theta)-\mathbb{E}_{f(\theta|\tau^{\ast})} \eta^{2}_{i}(\theta)+2\partial_{i} \log R(n\tau^{\ast})\left(\eta_{i}(\theta)-\mathbb{E}_{f(\theta|\tau^{\ast})} \eta_{i}(\theta)\right)\right]f_{T}(\tau^{\ast}).\label{Eq:WeightFactorsExponentialFamily}
\end{align}

Therefore,  as $|n\varepsilon|\rightarrow 0$, one expects that the following holds up to leading order term for the relative entropy
{\small\begin{align}
&H(P|P^{\varepsilon})(\tau^{\ast})=\frac{n^{4}}{8}\frac{1}{(q+2)^{2}}\mathbb{E}_{f(\theta|\tau^{\ast})}
\left(\sum_{i=1}^{q}\varepsilon_{i}^{2}\left(\eta^{2}_{i}(\theta)-\mathbb{E}_{f(\theta|\tau^{\ast})} \eta^{2}_{i}(\theta)+2\partial_{i} \log R(n\tau^{\ast})\left(\eta_{i}(\theta)-\mathbb{E}_{f(\theta|\tau^{\ast})} \eta_{i}(\theta)\right)
\right)\right)^{2}\nonumber\\
&\qquad\qquad+ \mathcal{O}(|n\varepsilon|^{6}).
\label{Eq:ExpressionForTheEllipse2a}
\end{align}}

Let us conclude the proof of the lemma  by justifying the validity of the expansion (\ref{Eq:ExpressionForTheEllipse2a}) when $|n\varepsilon|$ is small enough. A simple calculation of the Taylor series expansion of the density $f^{\varepsilon}(\theta|\tau^{*})$ shows that the remainder terms depend on terms of the form $\mathbb{E}_{f(\theta|\tau^{\ast})} \eta^{k}_{i}(\theta)$ for $k\in\mathbb{N}$ and on $\partial^{(k)}_{i} \log R(n\tau^{\ast})$. In regards to $\partial^{(k)}_{i} \log R(n\tau^{\ast})$ we have assumed Condition \ref{A:ConditionR}.



In regards to $\mathbb{E}_{f(\theta|\tau^{\ast})} \eta^{k}_{i}(\theta)$ we can be more explicit.  We use classical Laplace asymptotics and Condition \ref{A:ConditionR2} to show that such quantities will be uniformly bounded in $n$ for any $k\in \mathbb{N}$. The latter and Condition \ref{A:ConditionR} then imply that (\ref{Eq:ExpressionForTheEllipse2}) is a valid expansion. Indeed, notice that we can write for $j=1,\cdots, q$ and $k\in\mathbb{N}$
\begin{align}
\mathbb{E}_{f(\theta|\tau^{\ast})} \eta^{k}_{j}(\theta)&=\frac{\int \eta^{k}_{j}(\theta)e^{\sum_{i=1}^{q}\eta_{i}(\theta)\rho_{i}-\kappa A(\theta)}e^{n g(\theta)}d\theta}{\int e^{\sum_{i=1}^{q}\eta_{i}(\theta)\rho_{i}-\kappa A(\theta)}e^{n g(\theta)}d\theta}.\nonumber
\end{align}

We have assumed that the function g is smooth and has a maximum only at one interior non-degenerate critical point, say $\theta^{*}$. That means $\nabla g(\theta^{*})=0$ and $\text{det}\nabla^{2}g(\theta^{*})<0$ (Condition \ref{A:ConditionR2}). Then, classical Laplace asymptotics allows us to write as $n\rightarrow\infty$ up to leading order
\begin{align}
\mathbb{E}_{f(\theta|\tau^{\ast})} \eta^{k}_{j}(\theta)&\sim \frac{ \eta^{k}_{j}(\theta^{*})e^{\sum_{i=1}^{q}\eta_{i}(\theta^{*})\rho_{i}-\kappa A(\theta^{*})}e^{n g(\theta^{*})}\left(\frac{2\pi}{n}\right)^{q/2}\left[-\text{det}\nabla^{2}g(\theta^{*})\right]^{-1/2}}{e^{\sum_{i=1}^{q}\eta_{i}(\theta^{*})\rho_{i}-\kappa A(\theta^{*})}e^{n g(\theta^{*})}\left(\frac{2\pi}{n}\right)^{q/2}\left[-\text{det}\nabla^{2}g(\theta^{*})\right]^{-1/2}}\nonumber\\
&=\eta^{k}_{j}(\theta^{*}).\nonumber
\end{align}

Since the rest of the terms in the expansion are of lower order with respect to $n$, we can conclude that (\ref{Eq:ExpressionForTheEllipse2a}) is a valid expansion of the relative entropy even when $n$ gets large as long as $|n\varepsilon|$ is small. If $n$ is held fixed, then (\ref{Eq:ExpressionForTheEllipse2a}) is a valid expansion of the relative entropy if $|\varepsilon|$ is sufficiently small. This concludes the proof of the lemma.
\end{proof}

\section{Bias calculations}\label{S:BiasCalculations}

In this section we compute the bias of the estimator up to leading order with respect to $|\varepsilon|$. A similar computation has previously appeared in \cite{BarberVossWebster}. We present the result for the bias below, connecting it to the weight (\ref{Eq:weights}). In addition, the computations in \cite{BarberVossWebster} are for a ball as an acceptance region (\ref{Eq:Domain1}), whereas here we present the result for the more general case of an ellipse (\ref{Eq:Domain2}).

The bias of the estimator $\widehat{\pi_{X}[h]}_{\varepsilon,K}$ as defined by (\ref{Eq:Estimator}) is defined to be
\[
\text{bias}\left(\widehat{\pi_{X}[h]}_{\varepsilon,K}\right)=\mathbb{E}\widehat{\pi_{X}[h]}_{\varepsilon,K}-\pi_{X}[h].
\]

Algebraic computations similar to those in the case of a ball in \cite{BarberVossWebster} together with the computations of this section and the definition of the weights
by (\ref{Eq:weights}), shows that
\begin{align}
\text{bias}\left(\widehat{\pi_{X}[h]}_{\varepsilon,K}\right)&=\sum_{i=1}^{q}\varepsilon_{i}^{2}C_{i}(\tau^{*})+\mathcal{O}(|\varepsilon|^{3}),\label{Eq:Bias1}
\end{align}
where
\begin{align}
C_{i}(\tau^{\ast})&=\frac{1}{2(q+2)f(\tau^{*})}\mathbb{E}_{f(\theta|\tau^{\ast})}\left[ h(\theta) w_{i}(\tau^{\ast},\theta)\right].\label{Eq:Bias2}
\end{align}

It is clear that the last display reduces to the form presented in \cite{BarberVossWebster} if $\varepsilon_{i}=\varepsilon$ for every $i=1,\cdots, q$ and after substitution of the form of the weight factors $w_{i}(\tau^{\ast},\theta)$ from (\ref{Eq:weights}).

Let us now make specific the bias computations of (\ref{Eq:Bias1})-(\ref{Eq:Bias2}) for the case of an exponential family. By plugging in the expression for the weight $w_{i}(\tau^{*},\theta)$ from (\ref{Eq:WeightFactorsExponentialFamily}) in (\ref{Eq:Bias2}) we obtain
\begin{align}
C_{i}(\tau^{*})&=\frac{1}{2(q+2)f(\tau^{*})}\mathbb{E}_{f(\theta|\tau^{\ast})}\left[ h(\theta) w_{i}(\tau^{\ast},\theta)\right]\nonumber\\
&=\frac{n^{2}}{2(q+2)}\mathbb{E}_{f(\theta|\tau^{\ast})}\left[ h(\theta)\left(\eta^{2}_{i}(\theta)-\mathbb{E}_{f(\theta|\tau^{\ast})} \eta^{2}_{i}(\theta)+2\partial_{i} \log R(n\tau^{\ast})\left(\eta_{i}(\theta)-\mathbb{E}_{f(\theta|\tau^{\ast})} \eta_{i}(\theta)\right)\right)\right].\label{Eq:Bias2Exp}
\end{align}

Combined with (\ref{Eq:Bias1}), the latter expression shows that the typical order of the bias is $\mathcal{O}((n|\varepsilon|)^{2})$. However, as we will see in the specific examples of Section \ref{S:Example}, by taking into account the weight functions, specific observables of interest do behave better in terms of the $n$-dependence for the bias, and the bias may even be independent of $n$.

\section{Conditional mean rejection rate}\label{S:MeanRejectionRate}

The goal of this section is to compare the standard ABC algorithm that is using the ball as acceptance region (\ref{Eq:Domain1}) to the ABC algorithm that is using the ellipse as acceptance region (\ref{Eq:Domain2}). One possible way to do so, which can also be computed in practice, is the mean rejection rate per accepted particle for each one of the two algorithms. Let us define $R_{B}$ and $R_{E}$ to be the number of rejections before an accepted particle, conditional on the value of the observed statistic $\tau^{\ast}$, for ball and ellipse respectively. It is clear that both $R_{B}$ and $R_{E}$ are distributed according to geometric distribution
\[
R_{B}\sim \text{Geom}(p_{B}), \text{ and } R_{E}\sim \text{Geom}(p_{E}),
\]
where $p_{B}$ and $p_{E}$ being the acceptance probabilities for the case of acceptance region being a ball and an ellipse respectively. Writing $B_{\varepsilon}(\tau^{*})$ for the acceptance region for the ball, (\ref{Eq:Domain1}), and  $D_{\varepsilon}(\tau^{*})$ for the acceptance region for the ellipse, (\ref{Eq:Domain2}), we have respectively
\[
p_{B}=\mathbb{P}\left(\tau\in B_{\varepsilon}(\tau^{*})\right), \text{ and }p_{E}=\mathbb{P}\left(\tau\in D_{\varepsilon}(\tau^{*})\right),
\]
where we recall that in the case of a ball $\varepsilon$ is positive scalar, whereas in the case of an ellipse $\varepsilon\in\mathbb{R}^{q}_{+}$.

Now, we are interested in approximating in the small $|\varepsilon|$ regime the quantity
\[
U(\tau^{\ast})=\frac{\mathbb{E}\left[R_{E}|\tau^{\ast}\right]}{\mathbb{E} \left[R_{B}|\tau^{\ast}\right]}=\frac{p_{B}}{p_{E}}.
\]

We write $U(\tau^{\ast})=\frac{\mathbb{E}\left[R_{E}|\tau^{\ast}\right]}{\mathbb{E} \left[R_{B}|\tau^{\ast}\right]}$ in order to emphasize that $U$ is a random quantity depending on the observed statistic, hence it is data driven itself.

We have the following lemma, with proof presented at the end of this section.
\begin{lemma}\label{L:ConditionalRejectionRate}
Asymptotically, as $|\varepsilon|\rightarrow 0$ we have that
\begin{align}
U(\tau^{\ast})
&=\frac{\varepsilon^{q}}{\prod_{i=1}^{q}\varepsilon_{i}}\frac{f_{T}(\tau^{\ast}) + \varepsilon^{2}\frac{1}{2(q+2)}\Delta f_{T}(\tau^{\ast}) + \mathcal{O}(|\varepsilon|^{3})}{f_{T}(\tau^{\ast}) + \frac{1}{2(q+2)}\sum_{i=1}^{q}\varepsilon_{i}^{2}\partial_{\tau_{i}}^{2}f_{T}(\tau^{\ast}) + \mathcal{O}(|\varepsilon|^{3})}.\label{Eq:GeneralRatio}
\end{align}
\end{lemma}

Let us make now Lemma \ref{L:ConditionalRejectionRate}  precise for the case of exponential family. As we show in Section \ref{S:ExponentialFamily}, in the case of exponential family one sees clearly the dependence on the size of the data set, $n$, and as
a matter of fact, in the generic situation, the relation (\ref{Eq:GeneralRatio}) is then true if the products $n|\varepsilon|$ and $n\varepsilon$ are small.

\begin{corollary}\label{C:MeanRejectionRateExp}
 In the case of exponential family, (\ref{Eq:GeneralRatio}) reduces to
 \begin{align}
U(\tau^{\ast})&=\frac{\varepsilon^{q}}{\prod_{i=1}^{q}\varepsilon_{i}}\frac{1 + n^{2}\varepsilon^{2}\frac{1}{2(q+2)}\sum_{i=1}^{q}q_{i}(n,\tau^{\ast}) + \mathcal{O}(|n \varepsilon|^{3})}{1 + \frac{n^{2}}{2(q+2)}\sum_{i=1}^{q}\varepsilon_{i}^{2}q_{i}(n,\tau^{\ast}) + \mathcal{O}((n|\varepsilon|)^{3})},\nonumber
\end{align}
where
\begin{align}
q_{i}(n,\tau^{\ast})&=\frac{\partial^{2}_{i} R(n\tau^{\ast})}{R(n\tau^{\ast})}+\mathbb{E}_{f(\theta|\tau^{\ast})}\eta_{i}^{2}(\theta)
+2\partial_{i} \log R(n\tau^{\ast}) \mathbb{E}_{f(\theta|\tau^{\ast})}\eta_{i}(\theta).\label{Eq:WeightsRatio}
\end{align}
\end{corollary}

Interestingly  (\ref{Eq:GeneralRatio})  implies that if $\frac{\varepsilon_{i}}{\varepsilon}\gg 1$ for some directions $i=1,\cdots,q$ such that $\prod_{i=1}^{q}\frac{\varepsilon_{i}}{\varepsilon}>1$ then one can have benefits using the ellipse versus the ball since one can have that $U\ll 1$, which then implies that the mean rejection rate for the ellipse is  smaller than the mean rejection rate for the ball, leading to potential  computational gains. In Section \ref{S:Simulations} we  demonstrate this in a specific simulation study.

\begin{proof}[Proof of Lemma \ref{L:ConditionalRejectionRate}]
We denote $\tau_{\varepsilon}=(\tau_{1}\varepsilon_{1},\cdots,\tau_{q}\varepsilon_{q})$. Let us start with $p_{E}$. Using Taylor series expansion up to second order, we have for $|\varepsilon|$ sufficiently small (with all other parameters such as $n$ being fixed)
\begin{align}
p_{E}&=\mathbb{P}\left(\tau\in D_{\varepsilon}(\tau^{*})\right)=\int_{D_{\varepsilon}(\tau^{*})}f_{T}(\tau)d\tau=\prod_{i=1}^{q}\varepsilon_{i}\int_{B_{1}(0)}f_{T}(\tau_{\varepsilon}+\tau^{\ast})d\tau\nonumber\\
&=\prod_{i=1}^{q}\varepsilon_{i}\int_{B_{1}(0)}\left(f_{T}(\tau^{\ast})+\left<\tau_{\varepsilon},\nabla f_{T}(\tau^{\ast})\right>+\frac{1}{2}\tau_{\varepsilon}^{T}\nabla^{2}f_{T}(\tau^{\ast})\tau_{\varepsilon}+\mathcal{O}(|\varepsilon|^{3})\right)d\tau\nonumber\\
&=\prod_{i=1}^{q}\varepsilon_{i}\left[f_{T}(\tau^{\ast}) |B_{1}(0)|+ \frac{1}{2}\int_{B_{1}(0)}\left(\tau_{\varepsilon}^{T}\nabla^{2}f_{T}(\tau^{\ast})\tau_{\varepsilon}\right)d\tau + \mathcal{O}(|\varepsilon|^{3})\right]\nonumber\\
&=\prod_{i=1}^{q}\varepsilon_{i}\left[f_{T}(\tau^{\ast}) |B_{1}(0)|+ \frac{1}{2}\sum_{i=1}^{q}\varepsilon_{i}^{2}\partial_{\tau_{i}}^{2}f_{T}(\tau^{\ast}) \int_{B_{1}(0)}\tau_{i}^{2}d\tau + \mathcal{O}(|\varepsilon|^{3})\right]\nonumber\\
&=\prod_{i=1}^{q}\varepsilon_{i}\left[f_{T}(\tau^{\ast}) |B_{1}(0)|+ \frac{|B_{1}(0)|}{2(q+2)}\sum_{i=1}^{q}\varepsilon_{i}^{2}\partial_{\tau_{i}}^{2}f_{T}(\tau^{\ast}) + \mathcal{O}(|\varepsilon|^{3})\right],\nonumber
\end{align}
where in the last computation we use the symmetry of the domain and the relation
\[
\int_{B_{1}(0)}\tau_{i}^{2}d\tau=\frac{1}{q} \int_{B_{1}(0)}\left\|\tau\right\|^{2}d\tau=\frac{|B_{1}(0)|}{q+2}.
\]

It is clear that after setting $\varepsilon_{1}=\cdots=\varepsilon_{q}=\varepsilon$ we have
\begin{align}
p_{B}&=\mathbb{P}\left(\tau\in B_{\varepsilon}(\tau^{*})\right)\nonumber\\
&=\varepsilon^{q}\left[f_{T}(\tau^{\ast}) |B_{1}(0)|+ \varepsilon^{2}\frac{|B_{1}(0)|}{2(q+2)}\sum_{i=1}^{q}\partial_{\tau_{i}}^{2}f_{T}(\tau^{\ast}) + \mathcal{O}(|\varepsilon|^{3})\right]\nonumber\\
&=\varepsilon^{q}\left[f_{T}(\tau^{\ast}) |B_{1}(0)|+ \varepsilon^{2}\frac{|B_{1}(0)|}{2(q+2)}\Delta f_{T}(\tau^{\ast}) + \mathcal{O}(|\varepsilon|^{3})\right].\nonumber
\end{align}

Dividing the asymptotic expressions for $p_E$ and $p_B$ completes the proof of the lemma.
\end{proof}

\begin{proof}[Proof of Corollary \ref{C:MeanRejectionRateExp}]
Recalling now (\ref{Eq:DerivativeFactors}) and the identities (\ref{Eq:Zcombination1})-(\ref{Eq:Zcombination2}) we may compute
\begin{align}
\partial^{2}_{\tau_{i}}f_{T}(\tau)&=n^{2}\left[ \partial^{2}_{i} \log R(n\tau)-\partial^{2}_{i} \log Z(\rho+n\tau,\kappa+n)+ \left(\partial_{i} \log R(n\tau)-\partial_{i} \log Z(\rho+n\tau,\kappa+n)\right)^{2}\right]f_{T}(\tau)\nonumber\\
&=n^{2}\left[ \partial^{2}_{i} \log R(n\tau)-\partial^{2}_{i} \log Z(\rho+n\tau,\kappa+n)+ \left(\partial_{i} \log R(n\tau)\right)^{2}\right.\nonumber\\
&\left.\qquad+\left(\partial_{i} \log Z(\rho+n\tau,\kappa+n)\right)^{2}
-2\partial_{i} \log R(n\tau) \partial_{i} \log Z(\rho+n\tau,\kappa+n)\right]f_{T}(\tau)\nonumber\\
&=n^{2}\left[ \partial^{2}_{i} \log R(n\tau)+ \left(\partial_{i} \log R(n\tau)\right)^{2}+\mathbb{E}_{f(\theta|\tau)}\eta_{i}^{2}(\theta)
+2\partial_{i} \log R(n\tau) \mathbb{E}_{f(\theta|\tau)}\eta_{i}(\theta)\right]f_{T}(\tau)\nonumber\\
&=n^{2}\left[ \frac{\partial^{2}_{i} R(n\tau)}{R(n\tau)}+\mathbb{E}_{f(\theta|\tau)}\eta_{i}^{2}(\theta)
+2\partial_{i} \log R(n\tau) \mathbb{E}_{f(\theta|\tau)}\eta_{i}(\theta)\right]f_{T}(\tau).\nonumber
\end{align}

Recalling the definition of $q_{i}(n,\tau)$ from (\ref{Eq:WeightsRatio}) we then have
\begin{align}
\partial^{2}_{\tau_{i}}f_{T}(\tau)&=n^{2} q_{i}(n,\tau) f_{T}(\tau).\nonumber
\end{align}

Plugging the expression for $\partial^{2}_{\tau_{i}}f_{T}(\tau)$ into (\ref{Eq:GeneralRatio}) with $\tau=\tau^{\ast}$ and since by assumption $f_{T}(\tau^{\ast})\neq 0$ we conclude the proof of the Corollary.
\end{proof}


\section{Examples}\label{S:Example}

In this section, we present a few examples to illustrate the computations of this paper. In Section \ref{S:ExampleExponential} we consider the case where we want to infer the rate of an exponential distribution. In Section \ref{S:ExmapleNormal} we consider the case where we want to infer the mean and variance of a normal distribution. In both cases we use the ABC framework and we demonstrate how the relative entropy computations can help in determining optimally values for the threshold parameter, based on relative entropy considerations. Of course, in these cases, we have explicit information for the likelihood and thus ABC methods are not necessary for inference. However, having access to the formulas allow us to make explicit and informative computations in regards to the behavior of the ABC algorithm.

\subsection{Exponential distribution with unknown parameter}\label{S:ExampleExponential}
The goal of this section is to present a simple example where one can see how the computation of the weight factor can help in determining the tolerance parameter. In this example we will be inferring a scalar parameter, which implies that the $\varepsilon$ is scalar and the acceptance region is effectively a ball (\ref{Eq:Domain1}). We will see that even in this case, computation of the weight factor can be beneficial.

Let us assume that we have $n$ i.i.d $X_{1},\cdots,X_{n}\sim \text{Exp}(\theta)$ with unknown rate $\theta$. We then assume that the rate $\theta$ has gamma prior distribution, in particular we assume that $\theta\sim G(\alpha,\beta)$. Of course,  the computations of Section \ref{S:ExponentialFamily} carry over here. One can pick as sufficient statistic for the parameter $\theta$ to be $T=\bar{X}$ where $\bar{X}=\frac{1}{n}\sum_{i=1}^{n}X_{i}$ is the empirical mean.

With  this prior structure we have for the density posterior distribution
 \begin{align}
 f(\theta|\tau)&=\frac{(\beta+n\tau)^{\alpha+n}}{\Gamma(\alpha+n)}\theta^{\alpha+n-1}e^{-(\beta+n\tau)\theta},\label{Eq:PosteriorDensityExponential1}
 \end{align}
which is the density of gamma $G(\alpha+n,\beta+n\tau)$ distribution. Let us set $\alpha_{n}=\alpha+n$ and $\beta_{n}=\beta+n\tau$. In the notation of Section \ref{S:ExponentialFamily} we have $\eta(\theta)=-\theta$ and $R(n\tau)=\frac{1}{\Gamma(n)}(n\tau)^{n-1}$.

Recall that the observed value of the statistic is $\tau^{*}=\bar{X}$. We compute
\begin{align*}
&\text{Var}_{f(\theta|\tau^{\ast})}\left(\eta^{2}(\theta)\right)=\mathbb{E}_{f(\theta|\tau^{\ast})}
\theta^{4}-\left(\mathbb{E}_{f(\theta|\tau^{\ast})} \theta^{2}\right)^{2}\nonumber\\
&\quad=\frac{\prod_{i=1}^{4}(\alpha_{n}+i-1)}{\beta_{n}^{4}}
-\left( \frac{(\alpha_{n}+1)\alpha_{n}}{\beta_{n}^{2}}\right)^{2}
= \frac{\alpha_{n}(\alpha_{n}+1)(4\alpha_{n}+6)}{\beta_{n}^{4}},
\end{align*}

\begin{align*}
&\text{Var}_{f(\theta|\tau^{\ast})}\left(\eta(\theta)\right)=\frac{\alpha_{n}}{\beta_{n}^{2}}
\end{align*}
and
\begin{align*}
&\text{Cov}_{f(\theta|\tau^{\ast})}\left(\eta^{2}(\theta),\eta(\theta)\right)=-\mathbb{E}_{f(\theta|\tau^{\ast})}
\theta^{3}+\mathbb{E}_{f(\theta|\tau^{\ast})} \theta^{2}\mathbb{E}_{f(\theta|\tau^{\ast})} \theta\nonumber\\
&\quad=-\frac{\prod_{i=1}^{3}(\alpha_{n}+i-1)}{\beta_{n}^{3}}
+ \frac{(\alpha_{n}+1)\alpha^{2}_{n}}{\beta_{n}^{3}}= -2\frac{\alpha_{n}(\alpha_{n}+1)}{\beta_{n}^{3}}.
\end{align*}

Hence, we have the following approximation for the relative entropy
\begin{align}
&H(P|P^{\varepsilon})(\tau^{\ast})
= \frac{(n\varepsilon)^{4}}{72}\left[\text{Var}_{f(\theta|\tau^{\ast})}\left(
\eta^{2}(\theta)\right)+4\left(\partial_{i} \log R(n\tau^{\ast})\right)^{2}\text{Var}_{f(\theta|\tau^{\ast})}\left(
\eta_{i}(\theta)\right)\right.\nonumber\\
&\qquad\left.+2\partial_{i} \log R(n\tau^{\ast})\text{Cov}_{f(\theta|\tau^{\ast})}\left(\eta^{2}(\theta),\eta(\theta)\right)\right]+ \mathcal{O}(|n\varepsilon|^{6}) \nonumber\\
&=\frac{(n\varepsilon)^{4}}{72}\left[\frac{\alpha_{n}(\alpha_{n}+1)(4\alpha_{n}+6)}{\beta_{n}^{4}}+4\left(\frac{n-1}{n}\right)^{2}\frac{1}{\left(\tau^{\ast}\right)^{2}}\frac{\alpha_{n}}{\beta_{n}^{2}}
-4\frac{n-1}{n}\frac{1}{\tau^{*}}\frac{\alpha_{n}(\alpha_{n}+1)}{\beta_{n}^{3}}\right]+ \mathcal{O}(|n\varepsilon|^{6}).\nonumber
\end{align}

It is easy to see now that if $n$ is large, then
\[
\frac{\alpha_{n}(\alpha_{n}+1)(4\alpha_{n}+6)}{\beta_{n}^{4}}\approx \frac{4}{(\tau^{\ast})^{3}(\beta+n\tau^{\ast})}\approx \frac{1}{n}\frac{4}{(\tau^{\ast})^{4}}, \quad \frac{\alpha_{n}}{\beta_{n}^{2}}\approx \frac{1}{n}\frac{1}{(\tau^{\ast})^{2}}
\]
which then implies that for $n$ large such that $n\varepsilon$ is small enough
\begin{align}
H(P|P^{\varepsilon})(\tau^{\ast})
&=  n^{3}\varepsilon^{4} \frac{1}{18}\frac{1}{(\tau^{\ast})^{4}}+ \mathcal{O}(|n\varepsilon|^{6}),\nonumber
\end{align}
implying that  if there is a fixed tolerance level, say $\text{tol}$ such that we require $H(P|P^{\varepsilon})(\tau^{\ast})\leq\text{tol}$, then one would need to choose $\varepsilon$ such that
$n^{3}\varepsilon^{4} \frac{1}{18}\frac{1}{(\tau^{\ast})^{4}}\leq\text{tol}$. We remark here that the dependence on $n$ turned out to be better than $n^{4}$, which was due to the weight factor and that the value of the observed statistic affects the optimal choice of $\varepsilon$, in that the smaller it is, the smaller $\varepsilon$ should also become. The latter conclusion is also intuitive in the sense that the statistics $\tau^{\ast}=\bar{X}$ is the maximum likelihood estimator of $1/\theta$ and it measures the time that elapses on average between events. It is natural to expect that more frequent events would resort in smaller acceptance regions for ABC in order for the method to maintain its accuracy.

Let us conclude this section with the computation of the leading order of the bias. We choose as observable of interest the unknown rate $\theta$. We will show below that the leading order of the bias of this observable actually does not depend on $n$. For this purpose, we have for the specific example of exponential distribution
\begin{align}
C(\tau^{*})&=\frac{n^{2}}{6}\mathbb{E}_{f(\theta|\tau^{\ast})}\left[ \theta \left(\theta^{2}-\mathbb{E}_{f(\theta|\tau^{\ast})} \theta^{2}-2\frac{n-1}{n}\frac{1}{\tau^{*}}\left(\theta-\mathbb{E}_{f(\theta|\tau^{\ast})} \theta\right)\right)\right]\nonumber\\
&=\frac{n^{2}}{6}\left[\mathbb{E}_{f(\theta|\tau^{\ast})} \theta^{3} -\mathbb{E}_{f(\theta|\tau^{\ast})} \theta \mathbb{E}_{f(\theta|\tau^{\ast})} \theta^{2}-2\frac{n-1}{n}\frac{1}{\tau^{*}}\left(\mathbb{E}_{f(\theta|\tau^{\ast})} \theta^{2}-\left(\mathbb{E}_{f(\theta|\tau^{\ast})} \theta\right)^{2}\right)\right]\nonumber\\
&=\frac{n^{2}}{6}\left[\frac{\alpha_{n}(\alpha_{n}+1)(\alpha_{n}+2)}{\beta_{n}^{3}}-\frac{\alpha^{2}_{n}(\alpha_{n}+1)}{\beta^{3}_{n}}-2\frac{n-1}{n}\frac{1}{\tau^{*}}\frac{\alpha_{n}}{\beta^{2}_{n}}\right]\nonumber\\
&=\frac{n^{2}}{3}\frac{\alpha_{n}}{\beta_{n}^{2}}\left[\frac{\alpha_{n}+1}{\beta_{n}}-\frac{n-1}{n}\frac{1}{\tau^{*}}\right].\nonumber
\end{align}

Recalling the form of $\alpha_{n}$ and $\beta_{n}$ we then obtain that for large $n$,
\begin{align}
C(\tau^{*})&\approx\frac{n^{2}}{3}\frac{1}{n(\tau^{\ast})^{2}}\frac{1}{\tau^{\ast}}\frac{1}{n}\approx \frac{1}{3(\tau^{\ast})^{3}}.\nonumber
\end{align}

The latter implies that the bias does not depend on the size of the data set $n$ and it is of order $\varepsilon^{2}$. However, we also see that it is proportional to $\frac{1}{(\tau^{\ast})^{3}}$, which implies that if $\tau^{\ast}$ is small in value then the bias will be large for a fixed value of $\varepsilon$. The latter observation is in line with the qualitative dependence of the relative entropy with respect to the observed statistic $\tau^{*}$.

\subsection{Normal distribution with unknown mean and unknown variance.}\label{S:ExmapleNormal}

The goal of this section is to provide an explicit example where potential advantages appear when the acceptance region is an ellipse versus a ball, i.e., when one allows $\varepsilon_{i}$ to be different in each direction $i$.

We borrow the setup of the example from \cite{TurnerVanZandt2012}. In particular, we consider a hierarchical binomial model, where we want to infer the probability of correct responses to a signal detection experiment for each subject in a given group, see \cite{TurnerVanZandt2012}. Let $p_{i}$ be the probability of correct response of the i-th subject. Instead of modeling $p_{i}$ directly, we will model
\[
\text{logit}(p_{i})=\log\left(\frac{p_i}{1-p_i}\right).
\]

The $\text{logit}$ function is useful as it transforms $p_{i}\in(0,1)$ to $\text{logit}(p_{i})\in(-\infty,\infty)$. In particular, set $X_i=\text{logit}(p_{i})$ and assume that we have $n$ i.i.d $X_{1},\cdots,X_{n}\sim N(\mu,\sigma^{2})$ with unknown mean $\mu$ and unknown variance $\sigma^{2}$.

In addition, we mention here that normal distribution has been used as a simulator tool in ABC not only due to its analytic structure but also due to its empirically observed robustness for irregular probability distributions in chaotic or nearly chaotic simulation dynamics, see \cite{MeedsWelling2014, Wood2010}.

Let us denote by $\lambda=1/\sigma^{2}$ to be the precision parameter and assume the following normal-gamma prior for the pair $(\mu,\lambda)$
\begin{align}
\mu&\sim N\left(\mu_{0},\left(\kappa\lambda\right)^{-1}\right)\nonumber\\
\lambda&\sim G(\alpha,\beta), \quad \text{where G  is the gamma distribution}.\nonumber
\end{align}

Of course, normal distribution belongs to the exponential family, so the computations of Section \ref{S:ExponentialFamily} carry over here. One can pick as sufficient statistic for the parameter $\theta=(\mu,\lambda)$, $T=(\bar{X},S^{2})$ where $\bar{X}=\frac{1}{n}\sum_{i=1}^{n}X_{i}$ is the empirical mean and $S^{2}=\frac{1}{n-1}\sum_{i=1}^{n}(X_{i}-\bar{X})^{2}$ is the empirical variance. In this case we have $q=2$ for the range of the sufficient statistic $T$ and $m=2$ for the range of the parameter $\theta$. With  this prior structure we have for the density posterior distribution
 \begin{align}
 f(\theta|\tau)&=\frac{1}{Z(\alpha,\beta,\kappa)(2\pi)^{n/2}}\lambda^{1/2}\lambda^{\alpha+\frac{n}{2}-1}e^{-\lambda(\beta+\frac{n-1}{2}S^{2})}e^{-\frac{\lambda}{2}\left(\kappa(\mu-\mu_{0})^{2}+n(\bar{X}-\mu)^{2}\right)},\label{Eq:PosteriorDensityNormal1}
 \end{align}
where $Z(\alpha,\beta,\kappa)=\frac{\Gamma(\alpha)}{\beta^{\alpha}}\sqrt{\frac{2\pi}{\kappa}}$. Notice that if we define
\begin{align}
\mu_{n}&=\frac{\kappa\mu_{0}+n\bar{X}}{\kappa+n},\quad \beta_{n}=\beta+\frac{n-1}{2}S^{2}+\frac{\kappa n}{2(\kappa+n)}(\bar{X}-\mu_{0})^{2}, \quad \kappa_{n}=\kappa+n, \quad \alpha_{n}=\alpha+\frac{n}{2},\nonumber
\end{align}
then we can re-express the formula for the posterior density as a product of a normal and gamma density as follows (with some abuse of notation)
\begin{align}
f(\theta|\tau)&=f_{N(\mu_{n},(\kappa_{n}\lambda)^{-1})}(\mu)f_{G(\alpha_{n}, \beta_{n})} (\lambda).\label{Eq:PosteriorDensityNormal2}
\end{align}

The density for the marginal likelihood takes the form
\begin{align}
f_{T}(\tau)=\frac{Z(\alpha_{n},\beta_{n},\kappa_{n})}{Z(\alpha,\beta,\kappa)(2\pi)^{n/2}}=B_{n}\beta_{n}^{-\alpha_{n}}, \quad B_{n}=\frac{\Gamma(\alpha_{n})}{\Gamma(\alpha)}\sqrt{\frac{\kappa}{\kappa_{n}}}\beta^{\alpha}(2\pi)^{-n/2}.\nonumber
\end{align}

Let us next  rewrite the posterior density (\ref{Eq:PosteriorDensityNormal1}) in the canonical form. Writing $(n-1)S^{2}=n\frac{\sum_{i=1}^{n}X_{i}^{2}}{n}-n\bar{X}^{2}=n\overline{X^{2}}-n\bar{X}^{2}$, we see that when $\alpha=\frac{1+\kappa}{2}$ (\ref{Eq:PosteriorDensityNormal1}) takes the form (\ref{Eq:PosteriorExponential}) with the parameters
\begin{align}
\eta_{1}(\mu,\lambda)&=\mu\lambda, \quad \eta_{2}(\mu,\lambda)=-\frac{\lambda}{2}, \quad A(\mu,\lambda)=\frac{\mu^{2}\lambda}{2}-\frac{1}{2}\log\lambda,\nonumber
\end{align}
and
\begin{align}
\rho_{1}&=\kappa\mu_{0},\quad \rho_{2}=\kappa\mu_{0}^{2}+2\beta,\quad R(\tau)=\frac{1}{(2\pi)^{n/2}},\nonumber
\end{align}
where now we consider (equivalently) as sufficient statistic to be the pair $(\bar{X},\overline{X^{2}})$.

Let us now compute the weight functions $w_{i}(\tau,\theta)$ given by (\ref{Eq:weights}) or equivalently $\text{Var}_{f(\Theta|\tau^{\ast})}\left(
\eta^{2}_{i}(\theta)\right)$ and $\text{Cov}_{f(\Theta|\tau^{\ast})}\left(
\eta^{2}_{1}(\theta), \eta^{2}_{2}(\theta)\right)$ as they appear in (\ref{Eq:ExpressionForTheEllipse3}). This is the purpose of Corollary \ref{C:SpecialCalculationsNormalEx}.
\begin{corollary}\label{C:SpecialCalculationsNormalEx}
In the normal example studied in this section we have the following formulas
\begin{align}
\text{Var}_{f(\theta|\tau^{\ast})}\left(\eta^{2}_{1}(\theta)\right)&=\mu_{n}^{4}\frac{\prod_{i=1}^{4}(\alpha_{n}+i-1)}{\beta_{n}^{4}}+\frac{6\mu_{n}^{2}}{\kappa_{n}}\frac{\prod_{i=1}^{3}(\alpha_{n}+i-1)}{\beta_{n}^{3}}+
\frac{3}{\kappa_{n}^{2}}\frac{(\alpha_{n}+1)\alpha_{n}}{\beta_{n}^{2}}\nonumber\\
&\qquad-\left( \mu_{n}^{2}\frac{(\alpha_{n}+1)\alpha_{n}}{\beta_{n}^{2}} +\frac{1}{\kappa_{n}}\frac{\alpha_{n}}{\beta_{n}}\right)^{2}.\label{Eq:NormalWeight1}
\end{align}
\begin{align}
\text{Var}_{f(\theta|\tau^{\ast})}\left(\eta^{2}_{2}(\theta)\right)
&=\frac{1}{2^{4}}\left[\frac{\prod_{i=1}^{4}(\alpha_{n}+i-1)}{\beta_{n}^{4}}-\left(\frac{(\alpha_{n}+1)\alpha_{n}}{\beta_{n}^{2}}\right)^{2}\right].\label{Eq:NormalWeight2}
\end{align}
Lastly,
\begin{align}
\text{Cov}_{f(\theta|\tau^{\ast})}\left(\eta^{2}_{1}(\theta),\eta^{2}_{2}(\theta)\right)
&=\frac{1}{4}\left[\mu_{n}^{2}\frac{\alpha_{n}(\alpha_{n}+1)(4\alpha_{n}+6)}{\beta_{n}^{4}}+\frac{2}{\kappa_{n}}\frac{\alpha_{n}(\alpha_{n}+1)}{\beta_{n}^{3}}\right]
.\label{Eq:NormalWeight3}
\end{align}
\end{corollary}

Then, the following statement holds for the leading order term of the relative entropy.
\begin{corollary}\label{C:SpecialCalculationsRelEntropyNormalEx}
We have that
\begin{align}
H(P|P^{\varepsilon})(\tau^{\ast})
&=\frac{n^{4}}{128}\varepsilon_{1}^{4}\left[\mu_{n}^{4}\frac{\alpha_{n}(\alpha_{n}+1)(4\alpha_{n}+6)}{\beta_{n}^{4}}+\frac{\mu_{n}^{2}}{\kappa_{n}}\frac{\alpha_{n}(\alpha_{n}+1)(5\alpha_{n}+12)}{\beta_{n}^{3}}+
\frac{1}{\kappa_{n}^{2}}\frac{(2\alpha_{n}+3)\alpha_{n}}{\beta_{n}^{2}}\right]\nonumber\\
&+\frac{n^{4}}{128}\varepsilon_{2}^{4}\frac{1}{2^{4}} \frac{\alpha_{n}(\alpha_{n}+1)(4\alpha_{n}+6)}{\beta_{n}^{4}}+ \frac{n^{4}}{128}\varepsilon_{1}^{2}\varepsilon_{2}^{2}\frac{1}{2}\left[\mu_{n}^{2}\frac{\alpha_{n}(\alpha_{n}+1)(4\alpha_{n}+6)}{\beta_{n}^{4}}+\frac{2}{\kappa_{n}}\frac{\alpha_{n}(\alpha_{n}+1)}{\beta_{n}^{3}}\right]\nonumber\\
&+\mathcal{O}((n|\varepsilon|)^{6}).\label{Eq:RelEntropyNormal}
\end{align}
\end{corollary}

Proof of Corollaries \ref{C:SpecialCalculationsNormalEx} and \ref{C:SpecialCalculationsRelEntropyNormalEx} are deferred to the Appendix.

Now if the observed statistic $\tau^{\ast}=(\bar{X},\overline{X^{2}})$, or equivalently $\tau^{\ast}=(\bar{X},S^{2})$, is of order one, then it is easy to see that in terms of $n$ we have that $\beta_{n}=\mathcal{O}(n S^{2}/2)$, $\alpha_{n}=\mathcal{O}(n)$, $\kappa_{n}=\mathcal{O}(n)$ and $\mu_{n}=\mathcal{O}(\mu_{0}/n+\bar{X})$ as $n$ gets large. With some abuse of notation, and using Corollary \ref{C:SpecialCalculationsRelEntropyNormalEx}, these suggest that  for $n$ large and $\varepsilon$ sufficiently small such that $n|\varepsilon|$ is small we have up to leading order
\begin{align}
H(P|P^{\varepsilon})(\tau^{\ast})&\sim n^{4} \varepsilon_{1}^{4}  \left[\left(\frac{\mu^{4}_{0}}{n^{4}}+\bar{X}^{4}\right)\frac{1}{n (S^{2})^{4}} +\left(\frac{\mu_{0}^{2}}{n^{2}}+\bar{X}^{2}\right)\frac{1}{n (S^{2})^{3}}+\frac{1}{n^{2} (S^{2})^{2}}\right]+n^{4} \varepsilon_{2}^{4}\frac{1}{n (S^{2})^{4}}\nonumber\\
&\quad+n^{4}\varepsilon_{1}^{2}\varepsilon_{2}^{2}\left[\left(\frac{\mu^{2}_{0}}{n^{2}}+\bar{X}^{2}\right)\frac{1}{n (S^{2})^{4}} +\frac{1}{n^{2} (S^{2})^{3}}\right]\nonumber\\
&\sim  \varepsilon_{1}^{4}  \left[\left(\frac{\mu^{4}_{0}}{n}+n^{3}\bar{X}^{4}\right)\frac{1}{(S^{2})^{4}} +\left(n\mu_{0}^{2}+n^{3}\bar{X}^{2}\right)\frac{1}{(S^{2})^{3}}+\frac{n^{2}}{ (S^{2})^{2}}\right]+\varepsilon_{2}^{4}\frac{n^{3}}{(S^{2})^{4}}\nonumber\\
&\quad+\varepsilon_{1}^{2}\varepsilon_{2}^{2}\left[\left(n \mu^{2}_{0}+n^{3}\bar{X}^{2}\right)\frac{1}{(S^{2})^{4}} +\frac{1}{n^{2} (S^{2})^{3}}\right]\label{Eq:RelEntropyNormal2}.
 \end{align}

Expressions (\ref{Eq:RelEntropyNormal}) and (\ref{Eq:RelEntropyNormal2}) make it clear that the weighted effect on each direction $i=1,2$ is different. It then follows  that by exploiting the geometric structure of the relative entropy and the asymmetric effect of the sufficient statistic in different directions one may be able to afford to have larger acceptance region for the same overall tolerance level. Expression (\ref{Eq:RelEntropyNormal2}) also makes clear that the values of the observed sufficient statistics can have non-negligible consequences on the optimal choice of the threshold parameter. In Section \ref{S:Simulations}, we confirm these points via a simulation study.

\begin{remark}
We do point out that on this example we take advantage of the fact that, in a sense, mean and variance are independent directions. So allowing the possibility of choosing different $\varepsilon_i$ for each one of the directions $i=1,2$ is expected to be advantageous as opposed to forcing the choice $\varepsilon_1=\varepsilon_2=\varepsilon$. Notice though that their effect on the relative entropy is more complicated with the related cross-term being non-zero, see (\ref{Eq:NormalWeight3})-(\ref{Eq:RelEntropyNormal2}).
\end{remark}
\subsubsection{Mean rejection rate for the normal example}\label{SSS:MrejRateNormal}

In this subsection we make precise the mean rejection rate computations of Section \ref{S:MeanRejectionRate}. In this case, we have that $q_{i}(n,\tau^{\ast})$ from (\ref{Eq:WeightsRatio}) take the form
\begin{align}
q_{1}(n,\tau^{\ast})&=\mathbb{E}_{f(\theta|\tau^{\ast})}\eta_{1}^{2}(\theta)=\int\mu^{2}\lambda^{2}f(\mu,\lambda|\tau)d\mu d\lambda=\mu_{n}^{2}\frac{(\alpha_{n}+1)\alpha_{n}}{\beta_{n}^{2}} +\frac{1}{\kappa_{n}}\frac{\alpha_{n}}{\beta_{n}},\nonumber
\end{align}
and
\begin{align}
q_{2}(n,\tau^{\ast})&=\mathbb{E}_{f(\theta|\tau^{\ast})}\eta_{2}^{2}(\theta)=\frac{1}{4}\int \lambda^{2}f(\mu,\lambda|\tau)d\mu d\lambda=\frac{1}{4}\frac{(\alpha_{n}+1)\alpha_{n}}{\beta_{n}^{2}}.\nonumber
\end{align}


Thus, the leading order term of the ratio of the rejection rates  in the limit as $n|\varepsilon|$ and $n\varepsilon$ go to zero, takes the form
\begin{align}
U(\tau^{\ast})&=\frac{\mathbb{E}\left[R_{E}|\tau^{\ast}\right]}{\mathbb{E} \left[R_{B}|\tau^{\ast}\right]}=\frac{p_{B}}{p_{E}}\approx\frac{1}{\frac{\varepsilon_{1}\varepsilon_{2}}{\varepsilon^{2}}}\frac{
1+\frac{n^2\varepsilon^2}{8}\left[\mu_{n}^{2}\frac{(\alpha_{n}+1)\alpha_{n}}{\beta_{n}^{2}} +\frac{1}{\kappa_{n}}\frac{\alpha_{n}}{\beta_{n}} + \frac{1}{4}\frac{(\alpha_{n}+1)\alpha_{n}}{\beta_{n}^{2}}\right] }{1+\frac{n^2\varepsilon^{2}_{1}}{8}\left(\mu_{n}^{2}\frac{(\alpha_{n}+1)\alpha_{n}}{\beta_{n}^{2}} +\frac{1}{\kappa_{n}}\frac{\alpha_{n}}{\beta_{n}}\right)+ \frac{n^2\varepsilon^{2}_{2}}{8} \frac{1}{4}\frac{(\alpha_{n}+1)\alpha_{n}}{\beta_{n}^{2}}}.\label{Eq:RatioApproximationNormal}
\end{align}

It is important to note here however that this is an approximation formula valid only in the limit and thus care is needed in its application and interpretation. 
\subsubsection{Bias computations for the normal example}\label{SSS:BiasNormal}

Let us now discuss the behavior of the bias for the normal example. We will present the case of both mean and variance as observables. Let us start with the mean.

Let us set $h(\theta)=\mu$ and we want to compute the leading order of (\ref{Eq:Bias1}) for $h(\theta)=\mu$, i.e.,
\begin{align*}
\text{bias}(\hat{\mu})&=\sum_{i=1}^{q}\varepsilon_{i}^{2}C_{i}(\tau^{*})+\mathcal{O}(|\varepsilon|^{3}),
\end{align*}
where in this case we have
\begin{align}
C_{i}(\tau^{\ast})&=\frac{n^{2}}{8}\mathbb{E}_{f(\theta|\tau^{\ast})}\left[ \mu \left(\eta^{2}_{i}(\theta)-\mathbb{E}_{f(\theta|\tau^{\ast})} \eta^{2}_{i}(\theta)\right)\right].\label{Eq:Bias2normal}
\end{align}

The result is summarized in the following Corollary.
\begin{corollary}\label{C:Bias_mean_normal}
For $i=1,2$, $C_{i}(\tau^{\ast})$ from (\ref{Eq:Bias2normal}) satisfy
\begin{align}
C_{1}(\tau^{\ast})&=\frac{n^{2}}{4}\frac{\mu_{n}}{\kappa_{n}}\frac{\alpha_{n}}{\beta_{n}}, \text{ and } C_{2}(\tau^{\ast})=0.\label{Eq:Bias2normal_1}
\end{align}
Therefore, we have
\begin{align}
\text{bias}(\hat{\mu})&=\varepsilon_{1}^{2}C_{1}(\tau^{*})+\mathcal{O}(n|\varepsilon|^{3})\approx \varepsilon_{1}^{2}\frac{n}{4}\left(\frac{\mu_{0}}{n}+\bar{X}\right)\frac{1}{S^{2}}+\mathcal{O}((n|\varepsilon|)^{3}).\nonumber
\end{align}
\end{corollary}

Corollary \ref{C:Bias_mean_normal} shows that if $\mu_{0}$ and $\bar{X}$ are small in value, then $C_{1}$, up to leading order, does not depend on $n$, but it grows inversely proportionally to $S^{2}$. The latter behavior is qualitatively the same as the dependence of the relative entropy with respect to $S^{2}$.

Let us then investigate the variance. Recalling that the precision $\lambda=1/\sigma^{2}$, let us now set $h(\theta)=\frac{1}{\lambda}$. We  want to compute the leading order of (\ref{Eq:Bias1}) for $h(\theta)=\sigma^{2}=1/\lambda$, i.e.,
\begin{align*}
\text{bias}(\hat{\sigma^{2}})&=\sum_{i=1}^{q}\varepsilon_{i}^{2}C_{i}(\tau^{*})+\mathcal{O}((n|\varepsilon|)^{3}),
\end{align*}
where in this case we have
\begin{align}
C_{i}(\tau^{\ast})&=\frac{n^{2}}{8}\mathbb{E}_{f(\theta|\tau^{\ast})}\left[ \frac{1}{\lambda} \left(\eta^{2}_{i}(\theta)-\mathbb{E}_{f(\theta|\tau^{\ast})} \eta^{2}_{i}(\theta)\right)\right].\label{Eq:Bias2normalVar}
\end{align}

The result is summarized in the following Corollary.
\begin{corollary}\label{C:Bias_variance_normal}
For $i=1,2$, $C_{i}(\tau^{\ast})$ from (\ref{Eq:Bias2normalVar}) satisfy
\begin{align}
C_{1}(\tau^{\ast})&=\frac{n^{2}}{8}\left[ -2 \mu^{2}_{n}\frac{\alpha_n}{\beta_{n}}\frac{1}{\alpha_{n}-1}-\frac{1}{\kappa_{n}}\frac{1}{\alpha_{n}-1}\right], \text{ and } C_{2}(\tau^{\ast})=\frac{n^{2}}{32}\frac{\alpha_{n}}{\beta_{n}}\frac{-2}{\alpha_{n}-1}.\nonumber
\end{align}
Therefore, we have that for $n$ large enough and $\varepsilon$ small enough
\begin{align}
\text{bias}(\hat{\sigma^{2}})
&\approx \varepsilon_{1}^{2}\frac{n}{8}\left[-2\left(\frac{\mu_{0}}{n}+\bar{X}\right)^{2}\frac{1}{S^{2}}-\frac{1}{n}\right]+\varepsilon_{2}^{2}\left[-\frac{1}{16}\frac{n}{S^{2}}\right]+\mathcal{O}((n|\varepsilon|)^{3}).\nonumber
\end{align}
\end{corollary}

Proofs of Corollaries \ref{C:Bias_mean_normal} and \ref{C:Bias_variance_normal} are deferred to the appendix.

\section{Simulation studies}\label{S:Simulations}

In this section we present simulation studies to demonstrate some of the results of this paper. In particular, we show in the normal example of Section \ref{S:ExmapleNormal} that considering the ellipse instead of the circle as acceptance region, i.e., allowing $\varepsilon_{1}\neq\varepsilon_{2}$ has interesting advantages.

We present data for both the case of ball as acceptance region and for the case of an ellipse as acceptance region. In both cases, we fix a level of tolerance for the relative entropy, i.e., we require that $\varepsilon$ for the case of a ball or $\varepsilon_{1},\varepsilon_{2}$ for the case of an ellipse to be such that $H(P|P^{\varepsilon})(\tau^{\ast})\leq\text{tol}$. In the case of an ellipse, we choose $\varepsilon_{1}$ and $\varepsilon_{2}$ such that the area of the acceptance region, $\pi\varepsilon_{1}\varepsilon_{2}$, is maximized. In particular we solve the following optimization problem\footnote{In order to solve (\ref{Eq:OptimizationProblem3}) we used R's built-in function nloptr where the local and gradient-based optimization version of the Augmented Lagrangian algorithm was used.}
\begin{align}
&\textrm{maximize } (\pi\varepsilon_{1}\varepsilon_{2}) \textrm{ subject to leading order term in }(\ref{Eq:RelEntropyNormal})\leq\text{tol}
.\label{Eq:OptimizationProblem3}
\end{align}

The measure of efficiency will be the average number of rejections per accepted particle, which is an indicative measure of computation time needed to complete the algorithm for a predefined level of tolerance error in terms of relative entropy. In addition, in order to showcase the effect of the size of the data, $n$, we record values for different values of $n$.

In Section \ref{S:MeanRejectionRate} we derived asymptotic formulas for the mean rejection rates for both ellipse and ball. In the case of normal distribution, we derived that their ratio   $U(\tau^{\ast})$ asymptotically, as $n|\varepsilon|$ and $n\varepsilon$ are small, satisfies (\ref{Eq:RatioApproximationNormal}).

We write $U(\tau^{\ast})$ for $U$ to emphasize that the ratio is a function of the observed statistic and thus a random quantity by itself. In addition, we remark here however that this approximation is accurate only in the limit as $n|\varepsilon|\downarrow 0$.   To make clear that we are using the asymptotic formula we will write
\begin{align}
\tilde{U}(\tau^{\ast})&=
\frac{1}{\frac{\varepsilon_{1}\varepsilon_{2}}{\varepsilon^{2}}}\frac{
1+\frac{n^2\varepsilon^2}{8}\left[\mu_{n}^{2}\frac{(\alpha_{n}+1)\alpha_{n}}{\beta_{n}^{2}} +\frac{1}{\kappa_{n}}\frac{\alpha_{n}}{\beta_{n}} + \frac{1}{4}\frac{(\alpha_{n}+1)\alpha_{n}}{\beta_{n}^{2}}\right] }{1+\frac{n^2\varepsilon^{2}_{1}}{8}\left(\mu_{n}^{2}\frac{(\alpha_{n}+1)\alpha_{n}}{\beta_{n}^{2}} +\frac{1}{\kappa_{n}}\frac{\alpha_{n}}{\beta_{n}}\right)+ \frac{n^2\varepsilon^{2}_{2}}{8} \frac{1}{4}\frac{(\alpha_{n}+1)\alpha_{n}}{\beta_{n}^{2}}}.\nonumber
\end{align}

In all cases, the true values of the parameters under which the observed value of the statistic $\tau^{\ast}=(\bar{X},S^{2})$ is computed are $(\mu,\sigma^{2})=(0,1)$. The hyperparameters for the prior distributions are chosen to be $(\mu_{0},\kappa,\alpha,\beta)=(0,1,1,1)$. Also, in both cases the number of accepted sampled values, or particles, is fixed to $K=1000$. The goal is to make sure that the priors are appropriately chosen in order to concentrate the conclusions of the numerical studies on the effect of choosing a ball versus an ellipse as acceptance region.

%

Data for the ball as acceptance region, where $\varepsilon_{1}=\varepsilon_{2}=\varepsilon$ are presented in Table \ref{T:Ball}, whereas data for the ellipse as acceptance region, where $\varepsilon_{1}\neq\varepsilon_{2}$ are presented in Table \ref{T:Ellipse}. In the tables, we report the observed values of the sufficient statistic $\tau^{\ast}=(\bar{X},S^{2})$, the estimated values for $\mu$ and $\sigma^{2}$ along with their empirical standard deviations, as well as the average number of rejections per accepted particle, denoted by  $\hat{R}_{B}$ and $\hat{R}_{E}$ for ball and ellipse respectively. Clearly, the more rejections per accepted particle the algorithm has, the more time it takes to run for the same level of accuracy.  

\begin{table}[ht]
  \begin{tabular}[c]{|c|c|c|c|c|c|}
    \hline
     $(\text{tol},n)$ & $\tau^{\ast}=(\bar{X},S^{2})$ & $\varepsilon$ & $(\hat{\mu},\hat{\text{sd}}(\mu))$ & $(\hat{\sigma^{2}},\hat{\text{sd}}(\sigma^{2}))$  & $\hat{R}_{B}$  \\
    \hline
    $(0.05,100)$ & $(0.167,1.061)$ &  $0.055$ & $(0.161,0.109)$ & $(1.035,0.075)$ & $1484$  \\
    $(0.25,100)$ & $(-0.022,0.965)$ &  $0.083$ & $(-0.023,0.106)$ & $(0.987,0.076)$ & $608$  \\
    $(0.5,100)$ & $(0.061,1.099)$ &  $0.111$ & $(0.059,0.122)$ & $(1.053,0.079)$ & $384$  \\
    $(1,100)$ & $(-0.181,0.866)$ &  $0.096$ & $(-0.175,0.105)$ & $(0.936,0.071)$ & $389$  \\
    \hline
    $(0.05,300)$ & $(-0.112,1.003)$ &  $0.025$ & $(-0.111,0.061)$ & $(1.007,0.041)$ & $6998$  \\
    $(0.25,300)$ & $(0.022,0.974)$ &  $0.038$ & $(0.024,0.061)$ & $(0.989,0.041)$ & $2902$  \\
    $(0.5,300)$ & $(0.020,1.001)$ &  $0.046$ & $(0.021,0.062)$ & $(1.003,0.042)$ & $1982$  \\
    $(1,300)$ & $(-0.070,1.005)$  & $0.054$ & $(-0.069,0.063)$ & $(1.007,0.043)$ & $1490$  \\
    \hline
   $(0.05,600)$ & $(-0.014,1.025)$ &  $0.016$ & $(-0.013,0.004)$ & $(1.015,0.030)$ & $17645$  \\
    $(0.25,600)$ & $(0.023,0.933)$ &  $0.022$ & $(0.023,0.042)$ & $(0.968,0.028)$ & $8820$  \\
    $(0.5,600)$ & $(-0.009,0.986)$ &  $0.027$ & $(-0.010,0.041)$ & $(0.994,0.029)$ & $5805$  \\
    $(1,600)$ & $(0.016,0.972)$  & $0.003$ & $(0.014,0.004)$ & $(0.986,0.028)$ & $4172$  \\
    \hline
   $(0.05,1000)$ & $(-0.027,1.013)$ &  $0.011$ & $(-0.027,0.031)$ & $(1.017,0.023)$ & $35639$  \\
    $(0.25,1000)$ & $(0.016,0.971)$ &  $0.015$ & $(0.017,0.032)$ & $(0.986,0.022)$ & $17730$  \\
    $(0.5,1000)$ & $(-0.050,0.938)$ &  $0.0175$ & $(-0.051,0.032)$ & $(0.986,0.023)$ & $13868$  \\
    $(1,1000)$ & $(-0.012,0.995)$  & $0.022$ & $(-0.012,0.033)$ & $(0.992,0.023)$ & $8702$  \\
      \hline
\end{tabular}
\medskip
\caption{Data for ball as acceptance region.} \label{T:Ball}
\end{table}

\begin{table}[ht]
  \begin{tabular}[c]{|c|c|c|c|c|c|}
    \hline
     $(\text{tol},n)$ & $\tau^{\ast}=(\bar{X},S^{2})$ & $(\varepsilon_{1},\varepsilon_{2})$ & $(\hat{\mu},\hat{\text{sd}}(\mu))$ & $(\hat{\sigma^{2}},\hat{\text{sd}}(\sigma^{2}))$  & $\hat{R}_{E}$  \\
    \hline
    $(0.05,100)$ & $(0.167,1.061)$ &  $(0.065,0.049)$ & $(0.167,0.103)$ & $(1.037,0.074)$ & $1505$  \\
    $(0.25,100)$ & $(-0.022,0.965)$ &  $(0.155,0.069)$ & $(-0.027,0.121)$ & $(0.991,0.077)$ & $390$ \\
    $(0.5,100)$ & $(0.061,1.099)$ &  $(0.173,0.094)$ & $(0.063,0.133)$ & $(1.061,0.082)$ & $303$  \\
    $(1,100)$ & $(-0.081,0.866)$ &  $(0.113,0.086)$ & $(-0.178,0.106)$ & $(0.937,0.072)$ & $370$ \\
    \hline
    $(0.05,300)$ & $(-0.112,1.003)$ &  $(0.035,0.021)$ & $(-0.113,0.061)$ & $(1.004,0.042)$ & $5660$  \\
    $(0.25,300)$ & $(0.022,0.974)$ &  $(0.087,0.031)$ & $(0.194,0.071)$ & $(0.989,0.041)$ & $1559$  \\
    $(0.5,300)$ & $(0.020,1.001)$ &  $(0.106,0.038)$ & $(0.024,0.077)$ & $(0.998,0.040)$ & $1102$  \\
    $(1,300)$ & $(-0.070,1.005)$  & $(0.091,0.046)$ & $(-0.074,0.075)$ & $(1.006,0.045)$ & $1105$  \\
    \hline
   $(0.05,600)$ & $(-0.014,1.025)$ &  $(0.043,0.013)$ & $(-0.015,0.046)$ & $(1.014,0.028)$ & $7617$ \\
    $(0.25,600)$ & $(0.023,0.933)$ &  $(0.056,0.018)$ & $(0.022,0.049)$ & $(0.966,0.029)$ & $4045$ \\
    $(0.5,600)$ & $(-0.009,0.986)$ &  $(0.078,0.023)$ & $(-0.006,0.056)$ & $(0.991,0.029)$ & $2679$ \\
    $(1,600)$ & $(0.016,0.972)$  & $(0.087,0.026)$ & $(0.016,0.059)$ & $(0.986,0.031)$ & $1782$  \\
     \hline
   $(0.05,1000)$ & $(-0.027,1.033)$ &  $(0.028,0.009)$ & $(-0.026,0.035)$ & $(1.018,0.023)$ & $18043$  \\
    $(0.25,1000)$ & $(0.016,0.971)$ &  $(0.045,0.013)$ & $(0.017,0.039)$ & $(0.987,0.023)$ & $7085$  \\
    $(0.5,1000)$ & $(-0.050,0.938)$ &  $(0.036,0.015)$ & $(-0.011,0.048)$ & $(0.993,0.023)$ & $6270$ \\
    $(1,1000)$ & $(-0.012,0.985)$  & $(0.068,0.018)$ & $(-0.011,0.048)$ & $(0.993,0.023)$ & $3461$  \\
      \hline
\end{tabular}
\medskip
\caption{Data for ellipse as acceptance region.} \label{T:Ellipse}
\end{table}

Note that $\frac{\hat{R}_{E}}{\hat{R}_{B}}$  can be interpreted as the  per unit time computational gain that the algorithm with the ellipse as acceptance region has compared to the algorithm that uses the ball as acceptance region (i.e. the choice of two different distance metrics). Clearly, the smaller $\frac{\hat{R}_{E}}{\hat{R}_{B}}$ than $1$ is, the better the algorithm that uses the ellipse versus the algorithm that uses the ball for acceptance region is.

In Table \ref{T:Comparison} we see a comparison of the asymptotic $\tilde{U}$ and empirical $\frac{\hat{R}_{E}}{\hat{R}_{B}}$ for the data presented in Tables \ref{T:Ball} and \ref{T:Ellipse}.  It is clear that the asymptotic formula $\tilde{U}$ overestimates the ratio of mean rejection rates due to the fact that it is only an asymptotically correct formula. However, it is evident from the simulation results that  $\tilde{U}$ is of the correct order. 
In addition, for better presentation,  we have plotted in Figure \ref{Figure2}, the per unit time computational gain of the algorithm, $\frac{\hat{R}_{E}}{\hat{R}_{B}}$, for different values of tolerance levels taken from Table \ref{T:Comparison}.

\begin{table}[ht]
  \begin{tabular}[c]{|c||c|c|c|c||c|c|c|c|}
    \hline
     $(\text{tol},n)$ & $(0.05,100)$ & $(0.25,100)$ & $(0.5,100)$ & $(1,100)$  & $(0.05,300)$  & $(0.25,300)$ & $(0.5,300)$ & $(1,300)$\\
    \hline
    $\tilde{U}(\tau^{\ast})$ & $0.986$ &  $0.763$ & $0.884$ & $1.041$ & $0.947$  & $0.648$ & $0.651$ & $0.861$\\
    $\frac{\hat{R}_{E}}{\hat{R}_{B}}$ & $1.014$ &  $0.641$ & $0.790$ & $0.949$ & $0.809$ &  $0.537$ & $0.556$ & $0.741$\\
    \hline
 $(\text{tol},n)$ & $(0.05,600)$ & $(0.25,600)$ & $(0.5,600)$ & $(1,600)$  & $(0.05,1000)$  & $(0.25,1000)$ & $(0.5,1000)$ & $(1,1000)$\\
    \hline
    $\tilde{U}(\tau^{\ast})$ & $0.531$ &  $0.592$ & $0.540$ & $0.574$ & $0.592$  & $0.532$ & $0.574$ & $0.511$\\
    $\frac{\hat{R}_{E}}{\hat{R}_{B}}$ & $0.431$ &  $0.462$ & $0.461$ & $0.436$ & $0.506$ &  $0.401$ & $0.452$ & $0.397$\\
      \hline
\end{tabular}
\medskip
\caption{Asymptotic and Monte Carlo comparison for the two different choices of acceptance region (ball and ellipse).} \label{T:Comparison}
\end{table}

\begin{figure}[h]
\begin{center}
\includegraphics[scale=1, height=4cm, width=9cm, angle=0]{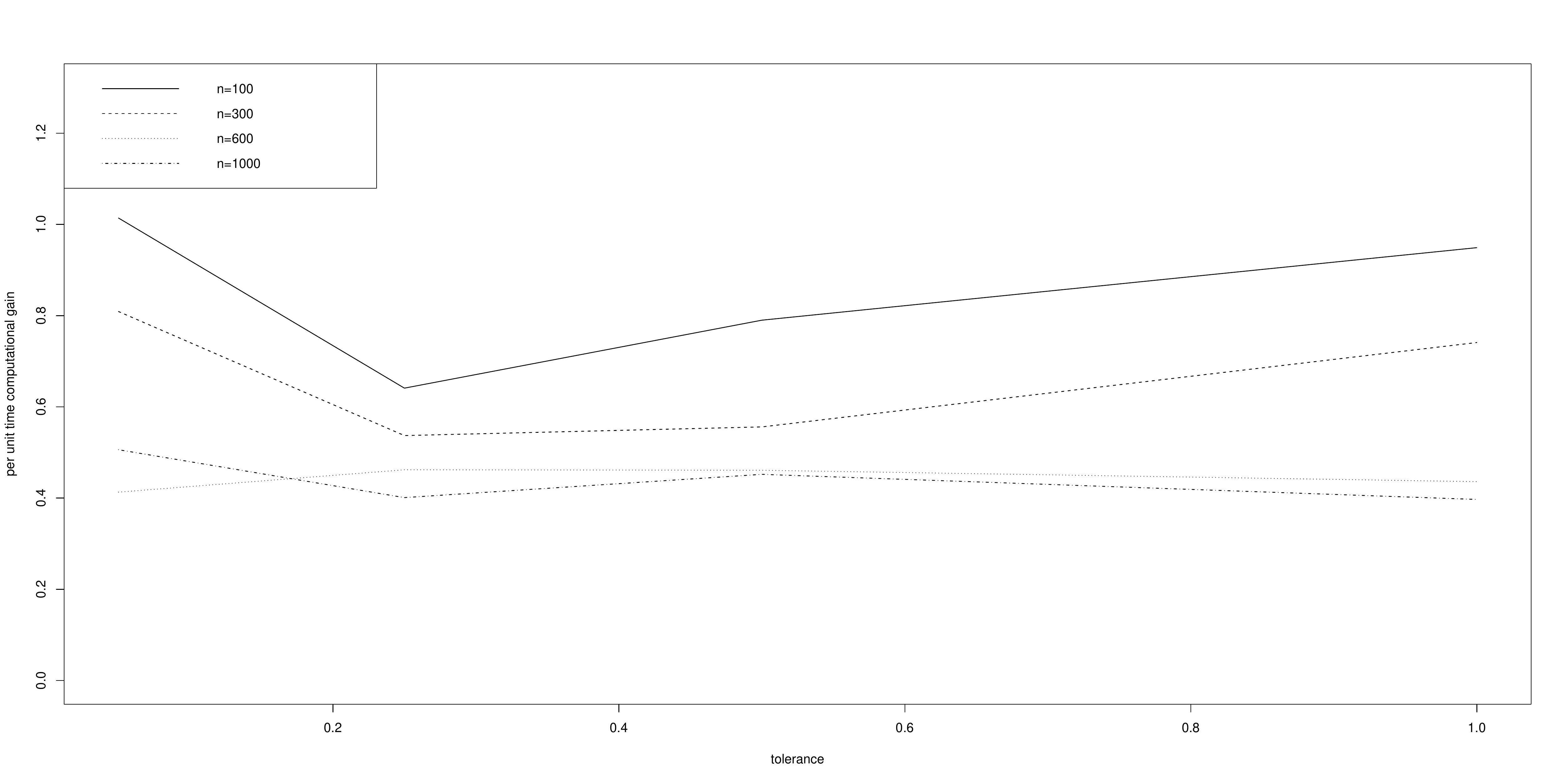}\hspace{2cm}
\caption{Per unit time computational gain of the algorithm, $\frac{\hat{R}_{E}}{\hat{R}_{B}}$. Values are from Table \ref{T:Comparison}. Computational gains are more evident as $n$ increases, but seem to stabilize for large $n$. }\label{Figure2}
\end{center}
\end{figure}

The conclusion is that for the same level of tolerance, allowing $\varepsilon_{1}\neq\varepsilon_{2}$ has  computational advantages as it reduces the number of rejections per accepted particle. 
In addition, the tables also demonstrate that as $n$ increases, $\varepsilon_{i}$ decrease. This is of course to be expected due to the form that the weight functions take in (\ref{Eq:RelEntropyNormal}). It is also noteworthy that the advantage of allowing $\varepsilon_{1}\neq\varepsilon_{2}$ seems to become more evident as $n$ increases even though the degree of the potential computational gain seems to diminish the larger the $n$ gets, see Figure \ref{Figure2} (at least for the simulation that is studied here).

\section{Conclusions}\label{S:Conclusions}
In this paper we analyzed the popular ABC algorithm via the lens of relative entropy ideas. By computing the leading order term in the expansion of relative entropy as function of the threshold vector parameter, we can quantify the effect of the chosen distance metric. In this direction, we showed that advantages arise when one exploits potential asymmetries in the distribution of observed statistics.  In particular,  one can then allow for a larger acceptance region which then implies a smaller rejection rate, which in turn results in faster convergence for the same level of tolerance error. This is part of a larger question, i.e., the quantification of the effect of distance metrics on the behavior of ABC. In this paper, we have seen that analyzing the ABC through the lens of relative entropy gives us a way to do so.

In addition, we characterized precisely the effect of the number of data points on the performance of ABC by showing that in order to maintain the same level of tolerance, one should decrease the threshold parameter as the number of data points increases. In connection to that, Figure 2 seems to suggest that as the number of data points, $n$, increase, the gain in performance due to the choice of the distance metric seems to stabilize. It would be of interest to theoretically address this aspect of the algorithm and as a consequence quantify in more precise terms the effect that the choice of the distance metric has as the number of data points increase.

One important conclusion of our analysis is that the leading order term in the expansion of relative entropy depends on certain weight factors that then govern the appropriate choice of threshold parameter. In this paper, our goal is to illustrate the phenomenon and for the numerical examples that we presented we could compute the weight functions in closed form. This was done in order to focus on the more theoretical aspects of the algorithm. For more general models, closed form computation may not be possible and one would have to result to simulation. The weight functions can be potentially monitored online and computed on the fly using Monte Carlo methods. We leave the development of the latter computational methods for future work.

\appendix

\section{Proof of intermediate results}\label{A:Appendix1}

\begin{lemma}\label{L:ZeroDerivative}
Assume Condition \ref{A:Regularity} and that the acceptance region $D_{\varepsilon}(\tau^{\ast})$ is given by (\ref{Eq:Domain2}). Let $f^{\varepsilon}(\theta|\tau^{\ast})$ be defined by (\ref{Eq:PerturbedDensity}). Then we have that $\nabla_{\varepsilon}f^{\varepsilon}(\theta|\tau^{\ast})|_{\varepsilon=0}=0$. In addition, we have that for all $i=1,\cdots, q$
\begin{align}
\partial^{2}_{\varepsilon_{i}}f^{\varepsilon}(\theta|\tau^{\ast})|_{\varepsilon=0}
&=\frac{1}{(q+2)f(\tau^{\ast})}\left(\partial_{\tau_{i}}^{2}f(\tau^{\ast},\theta)-
\frac{f(\tau^{\ast},\theta)}{f(\tau^{\ast})} \partial_{\tau_{i}}^{2}f(\tau^{\ast})\right),\label{Eq:HessianDiagonalDerivatives}
\end{align}
and for all $i,j=1,\cdots, q$ with $i\neq j$, $\partial^{2}_{\varepsilon_{i}\varepsilon_{j}}f^{\varepsilon}(\theta|\tau^{\ast})|_{\varepsilon=0}=0$.
\end{lemma}
\begin{proof}[Proof of Lemma \ref{L:ZeroDerivative}]
Let us denote $\tau_{\varepsilon}=(\tau_{1}\varepsilon_{1},\cdots,\tau_{q}\varepsilon_{q})$. Notice that by changing variables we can write
\begin{align}
f^{\varepsilon}(\theta|\tau^{\ast})&
=\frac{\int_{D_{\varepsilon}(\tau^{\ast})}f(\tau,\theta)d\tau}{\int_{D_{\varepsilon}(\tau^{\ast})}f(\tau)d\tau}=
\frac{\int_{D_{\varepsilon}(0)}f(\tau+\tau^{\ast},\theta)d\tau}{\int_{D_{\varepsilon}(0)}f(\tau+\tau^{\ast})d\tau}\nonumber\\
&=\frac{\int_{B_{1}(0)}f(\tau_{\varepsilon}+\tau^{\ast},\theta)\prod_{i=1}^{q}\varepsilon_{i}d\tau}{\int_{B_{1}(0)}f(\tau_{\varepsilon}+\tau^{\ast})\prod_{i=1}^{q}\varepsilon_{i}d\tau}
=\frac{\int_{B_{1}(0)}f(\tau_{\varepsilon}+\tau^{\ast},\theta)d\tau}{\int_{B_{1}(0)}f(\tau_{\varepsilon}+\tau^{\ast})d\tau}.\nonumber
\end{align}

Let $i=1,\cdots,q$. Taking now derivative with respect to $\varepsilon_{i}$, we obtain
\begin{align}
&\partial_{\varepsilon_{i}}f^{\varepsilon}(\tau^{\ast},\theta)
=\partial_{\varepsilon_{i}}\frac{\int_{B_{1}(0)}f(\tau_{\varepsilon}+\tau^{\ast},\theta)d\tau}{\int_{B_{1}(0)}f(\tau_{\varepsilon}+\tau^{\ast})d\tau}\nonumber\\
&=\frac{\int_{B_{1}(0)} \partial_{\tau_{i}}f(\tau_{\varepsilon}+\tau^{\ast},\theta) \tau_{i} d\tau\int_{B_{1}(0)}f(\tau_{\varepsilon}+\tau^{\ast})d\tau-
\int_{B_{1}(0)}f(\tau_{\varepsilon}+\tau^{\ast},\theta)d\tau\int_{B_{1}(0)} \partial_{\tau_{i}} f(\tau_{\varepsilon}+\tau^{\ast}) \tau_{i}d\tau}{\left(\int_{B_{1}(0)}f(\tau_{\varepsilon}+\tau^{\ast})d\tau\right)^{2}}.\nonumber
\end{align}

Due to the symmetry of the domain of integration around zero and because $s\mapsto s$ is an odd function, we  obtain when we evaluate at $\varepsilon=0$ for all $i=1,\cdots,q$
\begin{align}
\partial_{\varepsilon_{i}}f^{\varepsilon}(\theta|\tau^{\ast})|_{\varepsilon=0}
&=\frac{\left[\partial_{\tau_{i}}f(\tau^{\ast},\theta) f(\tau^{\ast})|B_{1}(0)|-
|B_{1}(0)|f(\tau^{\ast},\theta)\partial_{\tau_{i}} f(\tau^{\ast})\right]\int_{B_{1}(0)} \tau_{i} d\tau}{\left(|B_{1}(0)|f(\tau^{\ast})\right)^{2}}.\nonumber\\
&=0.
\end{align}

Hence in the case  of the ellipse, i.e. (\ref{Eq:Domain2}), we actually obtain that the order $O(|\varepsilon|^{2})$ and $O(|\varepsilon|^{3})$ vanish. We subsequently obtain for the relative entropy
\begin{align}
H(P|P^{\varepsilon})(\tau^{\ast})&=\frac{1}{8}\mathbb{E}_{f(\cdot|\tau^{\ast})}\left[\frac{\left|\varepsilon^{T}\nabla^{2}_{\varepsilon}f^{\varepsilon}(\theta|\tau^{\ast})|_{\varepsilon=0}\varepsilon\right|^{2}}{f^{2}(\theta|\tau^{\ast})}\right]+ \mathcal{O}(|\varepsilon|^{6}).\label{Eq:REmainterm}
\end{align}

Let us next compute the second derivative of the approximate posterior with respect to $\varepsilon$. We have for the diagonal terms of the Hessian $\nabla^{2}_{\varepsilon}f^{\varepsilon}(\theta|\tau^{\ast})$
\begin{align}
&\partial^{2}_{\varepsilon_{i}}f^{\varepsilon}(\theta|\tau^{\ast})
=\frac{\int_{B_{1}(0)} \partial^{2}_{\tau_{i}}f(\tau_{\varepsilon}+\tau^{\ast},\theta) \tau_{i}^{2}d\tau \int_{B_{1}(0)}f(\tau_{\varepsilon}+\tau^{\ast})d\tau-
\int_{B_{1}(0)}f(\tau_{\varepsilon}+\tau^{\ast},\theta)d\tau\int_{B_{1}(0)}\partial_{\tau_{i}}^{2} f(\tau_{\varepsilon}+\tau^{\ast}) \tau_{i}^{2} d\tau}{\left(\int_{B_{1}(0)}f(\tau_{\varepsilon}+\tau^{\ast})d\tau\right)^{2}}\nonumber\\
&-2\frac{\left(\int_{B_{1}(0)}\partial_{\tau_{i}}f(\tau_{\varepsilon}+\tau^{\ast},\theta) \tau_{i}d\tau\int_{B_{1}(0)}f(\tau_{\varepsilon}+\tau^{\ast})d\tau-
\int_{B_{1}(0)}f(\tau_{\varepsilon}+\tau^{\ast},\theta)d\tau\int_{B_{1}(0)}\partial_{\tau_{i}} f(\tau_{\varepsilon}+\tau^{\ast})\tau_{i}d\tau\right)}{\left(\int_{B_{1}(0)}f(\tau_{\varepsilon}+\tau^{\ast})d\tau\right)^{3}}\times\nonumber\\
&\hspace{1cm}\times \int_{B_{1}(0)}\partial_{\tau_{i}} f(\tau_{\varepsilon}+\tau^{\ast}) \tau_{i} d\tau.\nonumber
\end{align}

Evaluated at $\varepsilon=0$ we have
\begin{align}
\partial^{2}_{\varepsilon_{i}}f^{\varepsilon}(\theta|\tau^{\ast})|_{\varepsilon=0}
&=\frac{\int_{B_{1}(0)} \partial^{2}_{\tau_{i}}f(\tau^{\ast},\theta) \tau_{i}^{2}d\tau|B_{1}(0)|f(\tau^{\ast})-
|B_{1}(0)| f(\tau^{\ast},\theta) \int_{B_{1}(0)}\partial^{2}_{\tau_{i}} f(\tau^{\ast})\tau_{i}^{2} d\tau}{\left(|B_{1}(0)|f(\tau^{\ast})\right)^{2}}\nonumber\\
&=\frac{1}{|B_{1}(0)|f(\tau^{\ast})}\left(\int_{B_{1}(0)}\partial^{2}_{\tau_{i}}f(\tau^{\ast},\theta)\tau_{i}^{2}d\tau-
\frac{f(\tau^{\ast},\theta)}{f(\tau^{\ast})} \int_{B_{1}(0)}\partial^{2}_{\tau_{i}} f(\tau^{\ast}) \tau_{i}^{2}d\tau\right).\nonumber
\end{align}

Next using the symmetry of the domain of integration we compute
\[
\int_{B_{1}(0)}\tau_{i}^{2}d\tau=\frac{1}{q} \int_{B_{1}(0)}\left\|\tau\right\|^{2}d\tau=\frac{|B_{1}(0)|}{q+2}.
\]

Thus, we then have
\begin{align}
\partial^{2}_{\varepsilon_{i}}f^{\varepsilon}(\theta|\tau^{\ast})|_{\varepsilon=0}
&=\frac{1}{(q+2)f(\tau^{\ast})}\left(\partial_{\tau_{i}}^{2}f(\tau^{\ast},\theta)-
\frac{f(\tau^{\ast},\theta)}{f(\tau^{\ast})} \partial_{\tau_{i}}^{2}f(\tau^{\ast})\right),\label{Eq:HessianDiagonalDerivatives}
\end{align}

To complete the computations we also need the mixed-derivatives of the Hessian $\nabla^{2}_{\varepsilon}f^{\varepsilon}(\theta|\tau^{\ast})$. Hence, for $i,j=1,\cdots,q$ with $i\neq j$ we have
\begin{align}
&\partial^{2}_{\varepsilon_{i}\varepsilon_{j}}f^{\varepsilon}(\theta|\tau^{\ast})=\nonumber\\
&=\frac{\int_{B_{1}(0)} \partial^{2}_{\tau_{i}\tau_{j}}f(\tau_{\varepsilon}+\tau^{\ast},\theta) \tau_{i}\tau_{j}d\tau \int_{B_{1}(0)}f(\tau_{\varepsilon}+\tau^{\ast})d\tau-
\int_{B_{1}(0)}f(\tau_{\varepsilon}+\tau^{\ast},\theta)d\tau\int_{B_{1}(0)}\partial_{\tau_{i}\tau_{j}}^{2} f(\tau_{\varepsilon}+\tau^{\ast}) \tau_{i}\tau_{j} d\tau}{\left(\int_{B_{1}(0)}f(\tau_{\varepsilon}+\tau^{\ast})d\tau\right)^{2}}\nonumber\\
&+\frac{\int_{B_{1}(0)} \partial_{\tau_{i}}f(\tau_{\varepsilon}+\tau^{\ast},\theta) \tau_{i}d\tau \int_{B_{1}(0)}\partial_{\tau_{j}}f(\tau_{\varepsilon}+\tau^{\ast})\tau_{j}d\tau-
\int_{B_{1}(0)}\partial_{\tau_{j}}f(\tau_{\varepsilon}+\tau^{\ast},\theta)\tau_{j}d\tau\int_{B_{1}(0)}\partial_{\tau_{i}} f(\tau_{\varepsilon}+\tau^{\ast}) \tau_{i} d\tau}{\left(\int_{B_{1}(0)}f(\tau_{\varepsilon}+\tau^{\ast})d\tau\right)^{2}}\nonumber\\
&-2\frac{\left(\int_{B_{1}(0)}\partial_{\tau_{i}}f(\tau_{\varepsilon}+\tau^{\ast},\theta) \tau_{i}d\tau\int_{B_{1}(0)}f(\tau_{\varepsilon}+\tau^{\ast})d\tau-
\int_{B_{1}(0)}f(\tau_{\varepsilon}+\tau^{\ast},\theta)d\tau\int_{B_{1}(0)}\partial_{\tau_{i}} f(\tau_{\varepsilon}+\tau^{\ast})\tau_{i}d\tau\right)}{\left(\int_{B_{1}(0)}f(\tau_{\varepsilon}+\tau^{\ast})d\tau\right)^{3}}\times\nonumber\\
&\hspace{1cm}\times\int_{B_{1}(0)}\partial_{\tau_{j}} f(\tau_{\varepsilon}+\tau^{\ast}) \tau_{j} d\tau.\nonumber
\end{align}

Evaluated at $\varepsilon=0$ and using the symmetry of the domain and the fact that we are integrating odd polynomials we have
\begin{align}
\partial^{2}_{\varepsilon_{i}\varepsilon_{j}}f^{\varepsilon}(\theta|\tau^{\ast})|_{\varepsilon=0}
&=\frac{\int_{B_{1}(0)} \partial^{2}_{\tau_{i}\tau_{j}}f(\tau^{\ast},\theta) \tau_{i}\tau_{j}d\tau|B_{1}(0)|f(\tau^{\ast})-
|B_{1}(0)| f(\tau^{\ast},\theta) \int_{B_{1}(0)}\partial^{2}_{\tau_{i}\tau_{j}} f(\tau^{\ast})\tau_{i}\tau_{j} d\tau}{\left(|B_{1}(0)|f(\tau^{\ast})\right)^{2}}\nonumber\\
&=\frac{\int_{B_{1}(0)}\tau_{i}\tau_{j}d\tau}{|B_{1}(0)|f(\tau^{\ast})}\left(\partial^{2}_{\tau_{i}\tau_{j}}f(\tau^{\ast},\theta) -
\frac{f(\tau^{\ast},\theta)}{f(\tau^{\ast})}\partial^{2}_{\tau_{i}\tau_{j}} f(\tau^{\ast}) \right)\nonumber\\
&=0.\label{Eq:HessianOffDiagonalDerivatives}
\end{align}

Then, by (\ref{Eq:REmainterm}) and (\ref{Eq:HessianDiagonalDerivatives})-(\ref{Eq:HessianOffDiagonalDerivatives}) we conclude the proof of the lemma.
\end{proof}

\begin{proof}[Proof of Corollary \ref{C:SpecialCalculationsNormalEx}]
  We compute
\begin{align}
&\text{Var}_{f(\theta|\tau^{\ast})}\left(\eta^{2}_{1}(\theta)\right)=\mathbb{E}_{f(\theta|\tau^{\ast})}
\eta^{4}_{1}(\theta)-\left(\mathbb{E}_{f(\theta|\tau^{\ast})} \eta^{2}_{1}(\theta)\right)^{2}\nonumber\\
&\quad=\int\mu^{4}\lambda^{4}f(\mu,\lambda|\tau)d\mu d\lambda-\left(\int\mu^{2}\lambda^{2}f(\mu,\lambda|\tau)d\mu d\lambda\right)^{2}\nonumber\\
&\quad=\int\lambda^{4}\left(\mu_{n}^{4}+6\mu_{n}^{2}\frac{1}{\kappa_{n}\lambda}+3\frac{1}{\kappa_{n}^{2}\lambda^{2}}\right)f_{G(\alpha_{n}, \beta_{n})} (\lambda) d\lambda-\left(\int\lambda^{2}\left(\mu_{n}^{2}+\frac{1}{\kappa_{n}\lambda}\right)f_{G(\alpha_{n}, \beta_{n})} (\lambda) d\lambda\right)^{2}\nonumber\\
&\quad=\int\left(\mu_{n}^{4}\lambda^{4}+\frac{6\mu_{n}^{2}}{\kappa_{n}}\lambda^{3}+\frac{3}{\kappa_{n}^{2}}\lambda^{2}\right)f_{G(\alpha_{n}, \beta_{n})} (\lambda) d\lambda-\left(\int\left( \mu_{n}^{2}\lambda^{2} +\frac{1}{\kappa_{n}}\lambda\right)f_{G(\alpha_{n}, \beta_{n})} (\lambda) d\lambda\right)^{2}\nonumber\\
&\quad=\mu_{n}^{4}\frac{\prod_{i=1}^{4}(\alpha_{n}+i-1)}{\beta_{n}^{4}}+\frac{6\mu_{n}^{2}}{\kappa_{n}}\frac{\prod_{i=1}^{3}(\alpha_{n}+i-1)}{\beta_{n}^{3}}+
\frac{3}{\kappa_{n}^{2}}\frac{(\alpha_{n}+1)\alpha_{n}}{\beta_{n}^{2}}
-\left( \mu_{n}^{2}\frac{(\alpha_{n}+1)\alpha_{n}}{\beta_{n}^{2}} +\frac{1}{\kappa_{n}}\frac{\alpha_{n}}{\beta_{n}}\right)^{2}.\nonumber
\end{align}

Similarly, we have
\begin{align}
\text{Var}_{f(\theta|\tau^{\ast})}\left(\eta^{2}_{2}(\theta)\right)&=\mathbb{E}_{f(\theta|\tau^{\ast})}
\eta^{4}_{2}(\theta)-\left(\mathbb{E}_{f(\theta|\tau^{\ast})} \eta^{2}_{2}(\theta)\right)^{2}\nonumber\\
&=\frac{1}{2^{4}}\left[\int\lambda^{4}f(\mu,\lambda|\tau)d\mu d\lambda-\left(\int\lambda^{2}f(\mu,\lambda|\tau)d\mu d\lambda\right)^{2}\right]\nonumber\\
&=\frac{1}{2^{4}}\left[\frac{\prod_{i=1}^{4}(\alpha_{n}+i-1)}{\beta_{n}^{4}}-\left(\frac{(\alpha_{n}+1)\alpha_{n}}{\beta_{n}^{2}}\right)^{2}\right].\nonumber
\end{align}

In addition,
\begin{align}
&\text{Cov}_{f(\theta|\tau^{\ast})}\left(\eta^{2}_{1}(\theta),\eta^{2}_{2}(\theta)\right)=\mathbb{E}_{f(\theta|\tau^{\ast})}
\eta^{2}_{1}(\theta)\eta^{2}_{2}(\theta)-\mathbb{E}_{f(\theta|\tau^{\ast})} \eta^{2}_{1}(\theta)\mathbb{E}_{f(\theta|\tau^{\ast})} \eta^{2}_{2}(\theta)\nonumber\\
&=\frac{1}{4}\left[\int\mu^{2}\lambda^{4}f(\mu,\lambda|\tau)d\mu d\lambda-\int\mu^{2}\lambda^{2}f(\mu,\lambda|\tau)d\mu d\lambda \int\lambda^{2}f(\mu,\lambda|\tau)d\mu d\lambda\right]\nonumber\\
&=\frac{1}{4}\left[\mu_{n}^{2}\frac{\prod_{i=1}^{4}(\alpha_{n}+i-1)}{\beta_{n}^{4}}+\frac{1}{\kappa_{n}}\frac{\prod_{i=1}^{3}(\alpha_{n}+i-1)}{\beta_{n}^{3}}-\mu_{n}^{2}\frac{\alpha_{n}^{2}(\alpha_{n}+1)^{2}}{\beta_{n}^{4}}-
\frac{1}{\kappa_{n}}\frac{\alpha_{n}^{2}(\alpha_{n}+1)}{\beta_{n}^{3}}\right]\nonumber\\
&=\frac{1}{4}\left[\mu_{n}^{2}\frac{\alpha_{n}(\alpha_{n}+1)(4\alpha_{n}+6)}{\beta_{n}^{4}}+\frac{2}{\kappa_{n}}\frac{\alpha_{n}(\alpha_{n}+1)}{\beta_{n}^{3}}\right],\nonumber
\end{align}
concluding the proof of the Corollary.
\end{proof}

\begin{proof}[Proof of Corollary \ref{C:SpecialCalculationsRelEntropyNormalEx}]
Using Lemma \ref{L:RelEntropy_ExpFamilyGeneral} and Corollary \ref{C:SpecialCalculationsNormalEx} we have that
\begin{align}
&H(P|P^{\varepsilon})(\tau^{\ast})
= \frac{n^{4}}{128}\left[\sum_{i=1}^{2}\varepsilon_{i}^{4}\text{Var}_{f(\theta|\tau^{\ast})}\left(
\eta^{2}_{i}(\theta)\right)+ \varepsilon_{1}^{2}\varepsilon_{2}^{2}\text{Cov}_{f(\theta|\tau^{\ast})}\left(\eta^{2}_{1}(\theta),\eta^{2}_{2}(\theta)\right)\right]+\mathcal{O}(|\varepsilon|^{6}) \nonumber\\
&=\frac{n^{4}}{128}\varepsilon_{1}^{4}\left[\mu_{n}^{4}\frac{\prod_{i=1}^{4}(\alpha_{n}+i-1)}{\beta_{n}^{4}}+\frac{6\mu_{n}^{2}}{\kappa_{n}}\frac{\prod_{i=1}^{3}(\alpha_{n}+i-1)}{\beta_{n}^{3}}+
\frac{3}{\kappa_{n}^{2}}\frac{(\alpha_{n}+1)\alpha_{n}}{\beta_{n}^{2}}
-\left( \mu_{n}^{2}\frac{(\alpha_{n}+1)\alpha_{n}}{\beta_{n}^{2}} +\frac{1}{\kappa_{n}}\frac{\alpha_{n}}{\beta_{n}}\right)^{2}\right]\nonumber\\
&\qquad+\frac{n^{4}}{128}\varepsilon_{2}^{4}\frac{1}{2^{4}}\left[\frac{\prod_{i=1}^{4}(\alpha_{n}+i-1)}{\beta_{n}^{4}}-\left(\frac{(\alpha_{n}+1)\alpha_{n}}{\beta_{n}^{2}}\right)^{2}\right]\nonumber\\
&\qquad+
\frac{n^{4}}{128}\varepsilon_{1}^{2}\varepsilon_{2}^{2}\frac{1}{2}\left[\mu_{n}^{2}\frac{\alpha_{n}(\alpha_{n}+1)(4\alpha_{n}+6)}{\beta_{n}^{4}}+\frac{2}{\kappa_{n}}\frac{\alpha_{n}(\alpha_{n}+1)}{\beta_{n}^{3}}\right]+\mathcal{O}((n|\varepsilon|)^{6})\nonumber\\
&=\frac{n^{4}}{128}\varepsilon_{1}^{4}\left[\mu_{n}^{4}\frac{\alpha_{n}(\alpha_{n}+1)(4\alpha_{n}+6)}{\beta_{n}^{4}}+\frac{\mu_{n}^{2}}{\kappa_{n}}\frac{\alpha_{n}(\alpha_{n}+1)(5\alpha_{n}+12)}{\beta_{n}^{3}}+
\frac{1}{\kappa_{n}^{2}}\frac{(2\alpha_{n}+3)\alpha_{n}}{\beta_{n}^{2}}\right]\nonumber\\
&\qquad+\frac{n^{4}}{128}\varepsilon_{2}^{4}\frac{1}{2^{4}} \frac{\alpha_{n}(\alpha_{n}+1)(4\alpha_{n}+6)}{\beta_{n}^{4}}+ \frac{n^{4}}{128}\varepsilon_{1}^{2}\varepsilon_{2}^{2}\frac{1}{2}\left[\mu_{n}^{2}\frac{\alpha_{n}(\alpha_{n}+1)(4\alpha_{n}+6)}{\beta_{n}^{4}}+\frac{2}{\kappa_{n}}\frac{\alpha_{n}(\alpha_{n}+1)}{\beta_{n}^{3}}\right]+\mathcal{O}((n|\varepsilon|)^{6}),\nonumber
\end{align}
completing the proof of the Corollary.
\end{proof}

\begin{proof}[Proof of Corollary \ref{C:Bias_mean_normal}]
Recalling the form of $\eta_{1}(\theta)=\mu\lambda$ we have
\begin{align}
C_{1}(\tau^{\ast})&=\frac{n^{2}}{8}\mathbb{E}_{f(\theta|\tau^{\ast})}\left[ \mu \left(\mu^{2}\lambda^{2}-\mathbb{E}_{f(\theta|\tau^{\ast})} (\mu^{2}\lambda^{2})\right)\right]\nonumber\\
&=\frac{n^{2}}{8}\left[\mathbb{E}_{f(\theta|\tau^{\ast})}\left( \mu^{3}\lambda^{2}\right)-\mathbb{E}_{f(\theta|\tau^{\ast})}(\mu)\mathbb{E}_{f(\theta|\tau^{\ast})} (\mu^{2}\lambda^{2})\right]\nonumber\\
&=\frac{n^{2}}{8}\left[\int\lambda^{2}\left(\mu_{n}^{3}+3\frac{\mu_{n}}{\kappa_{n}\lambda}\right)f_{G(\alpha_{n}, \beta_{n})} (\lambda) d\lambda-\mu_{n} \left(\int\lambda^{2}\left(\mu_{n}^{2}+\frac{1}{\kappa_{n}\lambda}\right)f_{G(\alpha_{n}, \beta_{n})} (\lambda) d\lambda\right)\right]\nonumber\\
&=\frac{n^{2}}{8}\left[ 2\frac{\mu_{n}}{\kappa_{n}}\int\lambda f_{G(\alpha_{n}, \beta_{n})} (\lambda) d\lambda\right]\nonumber\\
&=\frac{n^{2}}{4}\frac{\mu_{n}}{\kappa_{n}}\frac{\alpha_{n}}{\beta_{n}}.\nonumber
\end{align}


Similar computations and recalling that $\eta_{2}(\theta)=-\frac{\lambda}{2}$ give
\begin{align}
C_{2}(\tau^{\ast})&=\frac{n^{2}}{32}\left[\mathbb{E}_{f(\theta|\tau^{\ast})} \left(\mu \lambda^{2}\right)-\mathbb{E}_{f(\theta|\tau^{\ast})} (\mu)\mathbb{E}_{f(\theta|\tau^{\ast})} (\lambda^{2})\right]\nonumber\\
&=\frac{n^{2}}{32}\left[\int\left(\mu_{n}\lambda^{2}f_{G(\alpha_{n}, \beta_{n})} (\lambda)\right) d\lambda
-\mu_{n}\int\left(\lambda^{2}f_{G(\alpha_{n}, \beta_{n})} (\lambda) \right) d\lambda\right]\nonumber\\
&=0.\nonumber
\end{align}
\end{proof}

\begin{proof}[Proof of Corollary \ref{C:Bias_variance_normal}]
Recalling the form of $\eta_{1}(\theta)=\mu\lambda$ we have
\begin{align}
C_{1}(\tau^{\ast})&=\frac{n^{2}}{8}\mathbb{E}_{f(\theta|\tau^{\ast})}\left[ \frac{1}{\lambda} \left(\mu^{2}\lambda^{2}-\mathbb{E}_{f(\theta|\tau^{\ast})} (\mu^{2}\lambda^{2})\right)\right]\nonumber\\
&=\frac{n^{2}}{8}\left[\mathbb{E}_{f(\theta|\tau^{\ast})}\left( \mu^{2}\lambda\right)-\mathbb{E}_{f(\theta|\tau^{\ast})}\left(\frac{1}{\lambda}\right)\mathbb{E}_{f(\theta|\tau^{\ast})} (\mu^{2}\lambda^{2})\right]\nonumber\\
&=\frac{n^{2}}{8}\left[\int\lambda\left(\mu_{n}^{2}+\frac{1}{\kappa_{n}\lambda}\right)f_{G(\alpha_{n}, \beta_{n})} (\lambda) d\lambda-\frac{\beta_{n}}{\alpha_n-1} \left(\int\lambda^{2}\left(\mu_{n}^{2}+\frac{1}{\kappa_{n}\lambda}\right)f_{G(\alpha_{n}, \beta_{n})} (\lambda) d\lambda\right)\right]\nonumber\\
&=\frac{n^{2}}{8}\left[ \mu^{2}_{n}\frac{\alpha_n}{\beta_{n}}+\frac{1}{\kappa_{n}}-\mu_{n}^{2}\frac{\alpha_{n}(\alpha_{n}+1)}{\beta_{n}(\alpha_{n}-1)}-\frac{1}{\kappa_{n}}\frac{\alpha_{n}}{\alpha_{n}-1}\right]\nonumber\\
&=\frac{n^{2}}{8}\left[ -2 \mu^{2}_{n}\frac{\alpha_n}{\beta_{n}}\frac{1}{\alpha_{n}-1}-\frac{1}{\kappa_{n}}\frac{1}{\alpha_{n}-1}\right].\nonumber 
\end{align}


Similarly, we have, for $n$ large
\begin{align}
C_{2}(\tau^{\ast})&=\frac{n^{2}}{32}\left[\mathbb{E}_{f(\theta|\tau^{\ast})} \left( \lambda\right)-\mathbb{E}_{f(\theta|\tau^{\ast})} \left(\frac{1}{\lambda}\right)\mathbb{E}_{f(\theta|\tau^{\ast})} (\lambda^{2})\right]\nonumber\\
&=\frac{n^{2}}{32}\frac{\alpha_{n}}{\beta_{n}}\left[1-\frac{\alpha_{n}+1}{\alpha_{n}-1}\right]\nonumber\\
&=\frac{n^{2}}{32}\frac{\alpha_{n}}{\beta_{n}}\frac{-2}{\alpha_{n}-1}.\nonumber
\end{align}

Therefore, we get that the bias of $\sigma^{2}=\frac{1}{\lambda}$ behaves for $n$ large enough and $\varepsilon$ small enough as follows to leading order
\begin{align}
\text{bias}(\hat{\sigma^{2}})&=\varepsilon_{1}^{2}C_{1}(\tau^{*})+\varepsilon_{2}^{2}C_{2}(\tau^{*})+\mathcal{O}((n|\varepsilon|)^{3})\nonumber\\
&\approx \varepsilon_{1}^{2}\frac{n}{8}\left[-2\left(\frac{\mu_{0}}{n}+\bar{X}\right)^{2}\frac{1}{S^{2}}-\frac{1}{n}\right]+\varepsilon_{2}^{2}\left[-\frac{1}{16}\frac{n}{S^{2}}\right]+\mathcal{O}((n|\varepsilon|)^{3}).\nonumber
\end{align}
\end{proof}


\begin{thebibliography}{99}

\bibitem{ArampatzisKatsReyBellet2016}
Georgios Arampatzis, Markos A. Katsoulakis and Luc Rey-Bellet,
Efficient estimators for likelihood ratio sensitivity indices of complex stochastic dynamics,
{\em J. Chem. Phys}, 144(10), (2016), pp. 104107.


\bibitem{ArampatzisKats2014}
Georgios Arampatzis and Markos A. Katsoulakis, Goal-oriented sensitivity analysis for lattice kinetic Monte Carlo simulations,{\em J. Chem. Phys.}, 140, (2014), pp. 124108.

\bibitem{BarberVossWebster} Stuart Barber, Jochen Voss and Mark Webster, The rate of convergence for approximate Bayesian computation, {\em     Electron. J. Statist.},     9(1), (2015), pp. 80-105.

\bibitem{Beaumont2002}
Mark A. Beaumont, Wenyang Zhang and David J Balding, Approximate Bayesian computation in population genetics, {\em Genetics}, 162(4), (2002), pp. 2025-2035.

\bibitem{Bertsekas2016}
Dimitri P. Bertsekas, {\em Nonlinear Programming}, 3rd edition, Athena Scientific, 2016.

\bibitem{BiauGuyader2015}
G\'{e}rard Biau, Fr\'{e}d\'{e}ric C\'{e}rou, and Arnaud Guyader, New insights into Approximate Bayesian Computation, {\em     Ann. Inst. H. Poincar\'{e} Probab. Statist.}, 51(1), (2015), pp. 376-403.

\bibitem{BlumFrancois2010}
Michael G. Blum and Olivier Francois,  Non-linear regression models for Approximate Bayesian Computation, {\em Statistics and Computing}, (20), (2010), pp. 63-73.

\bibitem{BlumTran2010}
Michael G. Blum and Viet Chi Tran, HIV with contact tracing: a case study in
approximate Bayesian computation, {\em Biostatistics}, 11(4), (2010), pp. 644-660.

\bibitem{Bortot2007}
P. Bortot, Stuart G. Coles, and Scott A. Sisson. Inference for stereological
extremes, {\em Journal of the American Statistical Association}, 102(477), (2007), pp. 84-92.

\bibitem{Boyd}
S. Boyd, and L. Vandenberghe, {\em Convex Optimization}, Cambridge University Press, 2004.

\bibitem{DupuisKatsPantazisPlechac2016}
Paul Dupuis, Markos A. Katsoulakis, Yannis Pantazis and Petr Plechac,  Path-space information bounds for uncertainty quantification and sensitivity analysis of stochastic dynamics, {\em SIAM J. Uncertainty Quantification}, 4, (2016), pp. 80-111.



\bibitem{MarinPudloRobertRyder2012}
Jean-Michel Marin, Pierre Pudlo, Christian P. Robert and Robin J. Ryder,  Approximate Bayesian computational methods, {\em Statistics and Computing}, 22(6), (2012), pp. 1167-1180.

\bibitem{MarjoramTavare2003}
Paul Marjoram, John Molitor, Vincent Plagnol, and Simon Tavar\'{e}, Markov
chain Monte Carlo without likelihoods, {\em Proceedings of the National Academy
of Sciences}, 100(26), (2003), 15324-15328.


\bibitem{MeedsWelling2014} Edward Meeds and Max Welling, GPS-ABC: Gaussian process surrogate approximate Bayesian computation, {\em Uncertainty in AI}, (2014).


\bibitem{KatsReyBelletWang2016}
 Markos A. Katsoulakis, Luc Rey-Bellet and Jie Wang, Scalable Information Inequalities for Uncertainty Quantification, {\ Journal of Computational Physics}, 336, (2017), pp. 513-545.

\bibitem{KatsPantazis2013}
Markos A. Katsoulakis and Yannis Pantazis, A Relative Entropy Rate Method for Path Space Sensitivity Analysis of Stationary Complex Stochastic Dynamics,{\em J. Chem. Phys.}, 138, (2013), pp. 054115.



\bibitem{Luenberger}
D. G. Luenberger, {\em Optimization by Vector Space Methods}, Wiley professional paperback series, 1969.


\bibitem{SottorivaTavare2010}
Andrea Sottoriva and Simon Tavar\'{e}. Integrating approximate Bayesian computation
with complex agent-based models for cancer research, {\em In Yves
Lechevallier and Gilbert Saporta, editors, Proceedings of COMPSTAT'2010},
Springer, (2010), pp. 57-66.

\bibitem{TanakaSisson2006}
Mark M. Tanaka, Andrew R. Francis, Fabio Luciani, and S. A. Sisson, Using
approximate Bayesian computation to estimate tuberculosis transmission
parameters from genotype data, {\em Genetics}, 173(3), (2006), 1511-1520.

\bibitem{TavareGriffits1997}
Simon Tavar\'{e}, David J. Balding, R. C. Griffiths, and Peter Donnelly,  Inferring
coalescence times from DNA sequence data,  {\em Genetics}, 145(2), (1997), pp. 505-518.

\bibitem{TurnerVanZandt2012}
Brandon M. Turner and Trisha Van Zandt, A tutorial on approximate Bayesian computation, {\em Journal of Mathematical Psychology}, 56, (2012), pp. 69-85.

\bibitem{Wilkinson2013}
Richard D. Wilkinson, Approximate Bayesian computation (ABC) gives exact
results under the assumption of model error, {\em Statistical Applications in
Genetics and Molecular Biology}, 12(2), (2013), pp. 129-141.

\bibitem{WilkinsonSteiperEtAl2011}
Richard D. Wilkinson, Michael E. Steiper, Christophe Soligo, Robert D. Martin,
Ziheng Yang, and Simon Tavar\'{e}, Dating primate divergences through an
integrated analysis of palaeontological and molecular data, {\em Systematic Biology},
60(1), (2011), pp. :16-31.

\bibitem{Wood2010}
Simon N. Wood, Statistical inference for noisy nonlinear ecological dynamic system, {\em Nature}, 466(7310), (2010), pp. 1102-1104.
\end{thebibliography}
\end{document}